\RequirePackage{snapshot}		
\RequirePackage[l2tabu,orthodox]{nag}		
\documentclass[reqno]{amsart}
\usepackage[margin=1.5in,bottom=1.25in,marginpar=.8in]{geometry}		

\usepackage{amsmath}		
\usepackage{amssymb}		
\usepackage{amsfonts}		
\usepackage{amsthm}		
\usepackage[foot]{amsaddr}		

\usepackage{mathtools}		

\mathtoolsset{%
}

\usepackage[utf8]{inputenc}		
\usepackage[T1]{fontenc}		

\usepackage[
cal=cm,
]
{mathalfa}

\usepackage{dsfont}		

\usepackage[proportional,tabular,lining,sf,mono=false]{libertine}
\usepackage{acronym}		
\renewcommand*{\aclabelfont}[1]{\acsfont{#1}}		
\newcommand{\acli}[1]{\emph{\acl{#1}}}		
\newcommand{\acdef}[1]{\emph{\acl{#1}} \textup{(\acs{#1})}\acused{#1}}		

\usepackage[labelfont={bf,small},labelsep=colon,font=small]{caption}	
\usepackage[dvipsnames,svgnames]{xcolor}		
\colorlet{MyRed}{Crimson!80!Black}
\colorlet{MyGreen}{DarkGreen!80!Black}
\colorlet{MyBlue}{MediumBlue}

\usepackage[]{titlesec}		
\titleformat{name=\section}{}{\thetitle.}{0.8em}{\centering\scshape}
\titleformat{name=\subsection}[runin]{}{\thetitle.}{0.5em}{\bfseries}[.]
\titleformat{name=\subsubsection}[runin]{}{\thetitle.}{0.5em}{\bfseries}[.]

\titleformat{name=\paragraph,numberless}[runin]{}{}{0em}{\bfseries}[]
\titlespacing{\paragraph}{0em}{\medskipamount}{1em}
\titleformat{name=\subparagraph,numberless}[runin]{}{}{0em}{}[.]
\titlespacing{\subparagraph}{0em}{0em}{0.5em}

\newcommand{\afterhead}{.\;}		
\newcommand{\para}[1]{\smallskip\paragraph{\textbf{#1\afterhead}}}

\usepackage{latexsym}		

\usepackage{subcaption}		
\usepackage{tikz}		
\usepackage{tikz-cd}		
\usetikzlibrary{calc,patterns,decorations.markings}		
\tikzset{negated/.style={
        decoration={markings,
            mark= at position 0.5 with {
                \node[transform shape,yshift=1pt] (tempnode) {\scriptsize$\xcancel{\quad}$};
            }
        },
        postaction={decorate}
    }
}

\usepackage{array}		
\usepackage{booktabs}		
\usepackage[inline,shortlabels]{enumitem}		
\setenumerate{itemsep=\smallskipamount,topsep=\medskipamount}		
\setitemize{itemsep=\smallskipamount,topsep=\medskipamount}		

\usepackage[kerning=true]{microtype}		

\usepackage{xspace}		

\usepackage[numbers,compress]{natbib}		

\bibpunct[, ]{[}{]}{,}{n}{,}{,}

\let\tempcite\cite
\newcommand{\citef}[2][]{\citeauthor{#2} \tempcite[#1]{#2}}
\renewcommand*{\cite}{\citef}		

\usepackage{hyperref}
\hypersetup{
colorlinks=true,
linktocpage=true,
pdfstartview=FitH,
breaklinks=true,
pdfpagemode=UseNone,
pageanchor=true,
pdfpagemode=UseOutlines,
plainpages=false,
bookmarksnumbered,
bookmarksopen=false,
bookmarksopenlevel=1,
hypertexnames=true,
pdfhighlight=/O,
urlcolor=MyBlue,linkcolor=MyBlue,citecolor=MyBlue,	
pdftitle={},
pdfauthor={},
pdfsubject={},
pdfkeywords={},
pdfcreator={pdfLaTeX},
pdfproducer={LaTeX with hyperref}
}

\usepackage[sort&compress,capitalize,nameinlink]{cleveref}		
\crefname{assumption}{Assumption}{Assumptions}

\usepackage{algorithm}		
\usepackage{algpseudocode}		

\usepackage{thmtools}		
\usepackage{thm-restate}		

\theoremstyle{plain}
\newtheorem{theorem}{Theorem}		
\newtheorem{corollary}{Corollary}		
\newtheorem{lemma}{Lemma}		
\newtheorem{proposition}{Proposition}		

\newtheorem*{corollary*}{Corollary}		

\theoremstyle{definition}
\newtheorem{definition}{Definition}		
\newtheorem{example}{Example}		

\newtheorem*{definition*}{Definition}		
\newtheorem*{assumption*}{Assumptions}		
\newtheorem*{example*}{Example}		

\theoremstyle{remark}
\newtheorem{remark}{Remark}		

\newtheorem*{remark*}{Remark}		

\def\endenv{\hfill{\small$\lozenge$}}

\renewcommand{\qedsymbol}{\small$\blacksquare$}		

\newcounter{proofpart}

\numberwithin{theorem}{section}		
\numberwithin{proposition}{section}		
\numberwithin{corollary}{section}		
\numberwithin{definition}{section}		
\numberwithin{remark}{section}		
\numberwithin{example}{section}		

\usepackage[yyyymmdd,hhmmss]{datetime}

\usepackage[showdeletions]{color-edits}		

\newcommand{\debug}[1]{#1}		

\newcommand{\newmacro}[2]{\newcommand{#1}{\debug{#2}}}		
\newcommand{\newop}[2]{\DeclareMathOperator{#1}{\debug{#2}}}		

\DeclarePairedDelimiter{\braces}{\{}{\}}		
\DeclarePairedDelimiter{\bracks}{[}{]}		
\DeclarePairedDelimiter{\parens}{(}{)}		

\DeclarePairedDelimiter{\abs}{\lvert}{\rvert}		

\DeclarePairedDelimiterX{\setdef}[2]{\{}{\}}{#1:#2}		
\DeclarePairedDelimiterXPP{\exclude}[1]{\mathopen{}\setminus}{\{}{\}}{}{#1}

\newcommand{\R}{\mathbb{R}}		

\DeclareMathOperator*{\argmax}{arg\,max}		

\DeclareMathOperator{\dist}{dist}		
\DeclareMathOperator{\dom}{dom}		
\newmacro{\Jac}{J}		
\DeclareMathOperator{\one}{\mathds{1}}		
\DeclareMathOperator{\relint}{ri}		
\DeclareMathOperator{\supp}{supp}		

\newcommand{\cf}{cf.\xspace}		
\newcommand{\eg}{e.g.,\xspace}		
\newcommand{\ie}{i.e.,\xspace}		

\newcommand{\textpar}[1]{\textup(#1\textup)}		

\newcommand{\alt}[1]{#1'}		
\newcommand{\altalt}[1]{#1''}		
\newcommand{\avg}[1][\state]{\bar#1}		

\newmacro{\dd}{\:d}		
\newcommand{\ddt}{\frac{d}{dt}}		
\newcommand{\eps}{\varepsilon}		
\newcommand{\pd}{\partial}		
\newcommand{\wilde}{\widetilde}		

\newmacro{\const}{c}		
\newmacro{\coef}{\lambda}		
\newmacro{\param}{\theta}		
\newmacro{\params}{\Theta}		

\newmacro{\pexp}{p}		
\newmacro{\qexp}{q}		
\newmacro{\rexp}{r}		

\newmacro{\beforestart}{0}		
\newmacro{\start}{1}		
\newmacro{\afterstart}{2}		
\newmacro{\running}{\start,\afterstart,\dotsc}		

\newmacro{\run}{n}		
\newmacro{\runalt}{k}		
\newmacro{\runaltalt}{m}		
\newmacro{\nRuns}{T}		
\newmacro{\runs}{\mathcal{\nRuns}}		

\newmacro{\state}{X}		
\newmacro{\statealt}{Y}		
\newmacro{\statealtalt}{Z}		

\newmacro{\aux}{\tilde\state}

\newcommand{\beforeinit}[1][\state]{\debug{#1}_{\beforestart}}		

\newcommand{\beforeiter}[1][\state]{\debug{#1}_{\runalt-1}}		
\newcommand{\iter}[1][\state]{\debug{#1}_{\runalt}}		

\newcommand{\prev}[1][\state]{\debug{#1}_{\run-1}}		
\newcommand{\curr}[1][\state]{\debug{#1}_{\run}}		
\renewcommand{\next}[1][\state]{\debug{#1}_{\run+1}}		

\newmacro{\tstart}{0}		
\renewcommand{\time}{\debug{t}}		
\newmacro{\timealt}{s}		
\newmacro{\horizon}{T}		

\newmacro{\traj}{x}		
\newmacro{\trajalt}{y}		
\newmacro{\trajaltalt}{z}		

\newmacro{\flowmap}{\Theta}		
\DeclarePairedDelimiterXPP{\flowof}[2]{\flowmap_{#1}}{(}{)}{}{#2}		

\newmacro{\lyap}{E}

\newop{\Eq}{Eq}		
\newop{\ESS}{ESS}		
\newop{\GESS}{GESS}		
\newop{\NSS}{NSS}		
\newop{\GNSS}{GNSS}		
\newop{\CE}{CE}		
\newop{\CCE}{CCE}		
\newop{\KKT}{KKT}		
\newop{\NI}{NI}		

\newop{\MVI}{MVI}		
\newop{\SVI}{SVI}		

\newop{\br}{\mathsf{br}}		
\newop{\BR}{\mathsf{BR}}		
\newop{\reg}{Reg}		
\newop{\preg}{\overline{Reg}}		
\newop{\val}{val}		

\newmacro{\weight}{\eps}		
\newcommand{\rBR}[1]{\BR_{#1}}		
\newcommand{\rEq}[1]{\Eq_{#1}}		
\newcommand{\rNash}[1][\weight]{\Eq_{#1}}		

\newmacro{\play}{i}		
\newmacro{\playalt}{j}		
\newmacro{\playaltalt}{k}		
\newmacro{\nPlayers}{N}		
\newmacro{\players}{\mathcal{I}}		

\newmacro{\pure}{\alpha}		
\newmacro{\purealt}{\beta}		
\newmacro{\purealtalt}{\gamma}		
\newmacro{\nPures}{A}		
\newmacro{\pures}{\mathcal{\nPures}}		

\newmacro{\strat}{x}		
\newmacro{\stratalt}{\alt\strat}		
\newmacro{\strataltalt}{\altalt\strat}		
\newmacro{\strats}{\mathcal{X}}		
\newmacro{\intstrats}{\strats^{\circle}}		

\newcommand{\eq}{\sol}		

\newmacro{\loss}{\ell}		
\newmacro{\meanloss}{L}		

\newmacro{\payvec}{v}		
\newmacro{\aggpay}{S}		

\newmacro{\lossvec}{w}		
\newmacro{\cumloss}{S}		
\newmacro{\score}{y}		

\newmacro{\pay}{u}		
\newmacro{\payv}{F}		
\newmacro{\pot}{P}		

\newmacro{\game}{\mathcal{G}}		
\newmacro{\gamefull}{\game(\pures,\payv)}		

\newmacro{\fingame}{\Gamma}		
\newmacro{\fingamefull}{\Gamma(\players,\pures,\loss)}		

\newmacro{\gmat}{g}		
\newmacro{\gdist}{\dist_{\gmat}}
\newmacro{\mfld}{M}		
\newmacro{\form}{\omega}		

\newmacro{\tvec}{z}		
\newmacro{\uvec}{u}		

\newmacro{\ball}{\mathbb{B}}		
\newmacro{\sphere}{\mathbb{S}}		

\newmacro{\vertex}{v}		
\newmacro{\vertexalt}{\alt\vertex}		
\newmacro{\vertexaltalt}{\altalt\vertex}		
\newmacro{\nVertices}{V}		
\newmacro{\vertices}{\mathcal{\nVertices}}		

\newmacro{\edge}{e}		
\newmacro{\edgealt}{\alt\edge}		
\newmacro{\edgealtalt}{\altalt\edge}		
\newmacro{\nEdges}{E}		
\newmacro{\edges}{\mathcal{\nEdges}}		

\newmacro{\graph}{\mathcal{G}}		
\newmacro{\graphall}{\graph(\vertices,\edges)}		

\newmacro{\vecspace}{\mathcal{V}}		
\newmacro{\subspace}{\mathcal{W}}		

\newmacro{\bvec}{e}		
\newmacro{\bvecs}{\mathcal{E}}		

\newmacro{\coord}{i}		
\newmacro{\coordalt}{j}		
\newmacro{\coordaltalt}{k}		
\newmacro{\nCoords}{d}		
\newmacro{\dims}{\nCoords}		
\newmacro{\vdim}{\nCoords}		

\newmacro{\pspace}{\mathcal{X}}		
\newmacro{\dspace}{\vecspace^{\ast}}		

\newmacro{\ppoint}{x}		
\newmacro{\ppointalt}{\alt\ppoint}		
\newmacro{\ppointaltalt}{\altalt\ppoint}		
\newmacro{\ppoints}{\mathcal{X}}		
\newmacro{\pstate}{X}		

\newmacro{\dpoint}{y}		
\newmacro{\dpointalt}{\alt\dpoint}		
\newmacro{\dpointaltalt}{\altalt\dpoint}		
\newmacro{\dpoints}{\mathcal{Y}}		
\newmacro{\dstate}{Y}		

\newmacro{\pvec}{u}		
\newmacro{\dvec}{v}		

\newmacro{\mat}{M}		
\newmacro{\hmat}{H}		

\newmacro{\ones}{\mathbf{1}}		
\newmacro{\eye}{I}		
\newmacro{\zer}{\mathbf{0}}		

\DeclarePairedDelimiter{\norm}{\lVert}{\rVert}		
\DeclarePairedDelimiterXPP{\dnorm}[1]{}{\lVert}{\rVert}{_{\ast}}{#1}		

\DeclarePairedDelimiterXPP{\onenorm}[1]{}{\lVert}{\rVert}{_{1}}{#1}		
\DeclarePairedDelimiterXPP{\twonorm}[1]{}{\lVert}{\rVert}{_{2}}{#1}		
\DeclarePairedDelimiterXPP{\supnorm}[1]{}{\lVert}{\rVert}{_{\infty}}{#1}		

\DeclarePairedDelimiterX{\braket}[2]{\langle}{\rangle}{#1,#2}		

\DeclarePairedDelimiterX{\inner}[2]{\langle}{\rangle}{#1,#2}		

\newcommand{\defeq}{\coloneqq}		

\newcommand{\from}{\colon}		
\newcommand{\too}{\rightrightarrows}		

\newmacro{\source}{O}		
\newmacro{\sink}{D}		

\newmacro{\pair}{i}		
\newmacro{\pairalt}{j}		
\newmacro{\pairaltalt}{k}		
\newmacro{\nPairs}{N}		
\newmacro{\pairs}{\mathcal{\nPairs}}		

\newmacro{\route}{p}		
\newmacro{\routealt}{\alt\route}		
\newmacro{\routealtalt}{\altalt\route}		
\newmacro{\nRoutes}{P}		
\newmacro{\routes}{\mathcal{\nRoutes}}		

\newmacro{\flow}{f}		
\newmacro{\flowalt}{\alt\flow}		
\newmacro{\flowaltalt}{\altalt\flow}		
\newmacro{\flows}{\mathcal{F}}		

\newmacro{\load}{x}		
\newmacro{\loadalt}{\alt\load}		
\newmacro{\loadaltalt}{\altalt\load}		
\newmacro{\loads}{\mathcal{X}}		

\newop{\Opt}{Opt}		
\newop{\Sol}{Sol}		
\newop{\gap}{Gap}		
\newop{\orcl}{Or}		

\newmacro{\obj}{f}		
\newmacro{\objalt}{g}		
\newmacro{\sobj}{F}		

\newmacro{\gvec}{g}		
\newmacro{\oper}{M}		

\newcommand{\sol}[1][\point]{#1^{\ast}}		

\newmacro{\vbound}{G}		
\newmacro{\lips}{L}		
\newmacro{\strong}{\alpha}		
\newmacro{\smooth}{\beta}		

\newop{\tspace}{T}		
\newop{\tcone}{TC}		
\newop{\dcone}{\tcone^{\ast}}		
\newop{\ncone}{NC}		
\newop{\pcone}{PC}		
\newop{\hull}{\Delta}		

\newmacro{\cvx}{\mathcal{C}}		
\newmacro{\subd}{\partial}		
\newmacro{\subsel}{\wilde\nabla}		

\newmacro{\minmax}{L}		

\newmacro{\minvar}{\theta}		
\newmacro{\minvaralt}{\alt\minvar}		
\newmacro{\minvars}{\Theta}		

\newmacro{\maxvar}{\phi}		
\newmacro{\maxvaralt}{\alt\maxvar}		
\newmacro{\maxvars}{\Phi}		

\newop{\Eucl}{\Pi}		
\newop{\logit}{\Lambda}		
\newop{\dkl}{KL}		

\newmacro{\hreg}{h}		
\newmacro{\hconj}{\hreg^{\ast}}		
\newmacro{\breg}{D}		
\newmacro{\mprox}{P}		
\newmacro{\mirror}{\nabla\hconj}		
\newcommand{\fench}[1]{\debug{F}_{#1}}		
\newmacro{\hstr}{K}		
\newmacro{\depth}{H}		
\newmacro{\proxdom}{\strats_{\hreg}}		

\DeclarePairedDelimiterXPP{\proxof}[2]{\mprox_{#1}}{(}{)}{}{#2}		

\newmacro{\zone}{\mathcal{D}}		
\newmacro{\dzone}{\mathcal{F}}		

\newmacro{\point}{x}		
\newmacro{\pointalt}{\alt\point}		
\newmacro{\pointaltalt}{\altalt\point}		
\newmacro{\points}{\mathcal{X}}		
\newmacro{\intpoints}{\relint\points}		

\newmacro{\base}{p}		
\newmacro{\basealt}{q}		
\newmacro{\basealtalt}{u}		

\newmacro{\open}{\mathcal{U}}		
\newmacro{\closed}{\mathcal{C}}		
\newmacro{\cpt}{\mathcal{K}}		
\newmacro{\nhd}{\mathcal{U}}		

\newop{\ex}{\mathbb{E}}		
\newop{\prob}{\mathbb{P}}		
\newop{\Var}{Var}		
\newop{\simplex}{\hull}		

\DeclarePairedDelimiterXPP{\exof}[1]{\ex}{[}{]}{}{
 #1}

\DeclarePairedDelimiterXPP{\probof}[1]{\prob}{(}{)}{}{
 #1}

\DeclarePairedDelimiterXPP{\oneof}[1]{\one}{\{}{\}}{}{
 #1}

\newmacro{\sample}{\omega}		
\newmacro{\samples}{\Omega}		

\newmacro{\filter}{\mathcal{F}}		
\newmacro{\probspace}{(\samples,\filter,\prob)}		

\newmacro{\event}{E}       
\newmacro{\eventalt}{H}       
\newmacro{\mean}{\mu}		
\newmacro{\sdev}{\sigma}		
\newmacro{\variance}{\sdev^{2}}		

\newmacro{\signal}{V}		
\newmacro{\step}{\gamma}		
\newmacro{\learn}{\eta}		
\newmacro{\diff}{\delta}		

\newmacro{\proper}{\tau}		

\newmacro{\error}{Z}		
\newmacro{\noise}{U}		
\newmacro{\bias}{b}		
\newmacro{\brown}{W}		

\newmacro{\serror}{\theta}		
\newmacro{\snoise}{\xi}		
\newmacro{\sbias}{\psi}		

\newmacro{\sbound}{M}		
\newmacro{\bbound}{B}		
\newmacro{\noisepar}{\sdev}		
\newmacro{\noisevar}{\variance}		

\addauthor[Saeed]{SH}{DarkGreen}

\addauthor[Rida]{RL}{Indigo}

\addauthor[Pan]{PM}{MediumBlue}

\newop{\loc}{loc}

\newcommand{\choice}[1]{\debug{Q}_{#1}}		
\newmacro{\leb}{\mu}

\newmacro{\contclass}{\mathsf{Pot}\cup\mathsf{Mon}^{\ast}}

\addauthor[SyS]{SyS}{OliveGreen}

\begin{document}


\newcommand{\longtitle}{\uppercase{Learning in nonatomic games, Part I:\\Finite action spaces and population games}}		
\newcommand{\runtitle}{\uppercase{Learning in nonatomic games with finite action spaces}}		

\title[\runtitle]{\longtitle}		

\author[S.~Hadikhanloo]{Saeed Hadikhanloo$^{\ast}$}		
\address{$^{\ast}$ Tweag I/O, Paris, France.}		
\email{saeed.hadikhanloo@tweag.io}		

\author[R.~Laraki]{Rida Laraki$^{\dag}$}		
\address{$^{\dag}$ CNRS (Lamsade, University of Dauphine-PSL) and University of Liverpool (Computer Science Department).}		
\email{rida.laraki@dauphine.fr}		

\author
[P.~Mertikopoulos]
{\\Panayotis Mertikopoulos$^{\sharp}$}
\address{$^{\sharp}$\,%
Univ. Grenoble Alpes, CNRS, Inria, Grenoble INP, LIG, 38000, Grenoble, France.}
\email{panayotis.mertikopoulos@imag.fr}

\author[S.~Sorin]{Sylvain Sorin$^{\ddag}$}		
\address{$^{\ddag}$ Institut  de  Mathématiques Jussieu-PRG, Sorbonne Université, UPMC, CNRS UMR 7586.}		
\email{sylvain.sorin@imj-prg.fr}		




\newacro{LHS}{left-hand side}
\newacro{RHS}{right-hand side}
\newacro{iid}[i.i.d.]{independent and identically distributed}
\newacro{lsc}[l.s.c.]{lower semi-continuous}

\newacro{CFP}[cFP]{continuous fictitious play}
\newacro{cRFP}[cFP$_{\weight}$]{continuous $\weight$-regularized fictitious play}
\newacro{CDA}[cDA]{continuous dual averaging}

\newacro{BRD}{best reply dynamics}
\newacro{RBRD}{regularized best reply dynamics}
\newacro{VBRD}{vanishingly regularized best reply dynamics}
\newacro{DAD}[DAD]{dual averaging dynamics}

\newacro{APT}{asymptotic pseudotrajectory}
\newacroplural{APT}{asymptotic pseudotrajectories}
\newacro{KKT}{Karush\textendash Kuhn\textendash Tucker}
\newacro{ESS}{evolutionarily stable state}
\newacro{GESS}{globally evolutionarily stable state}
\newacro{NSS}{neutrally stable state}
\newacro{GNSS}{globally neutrally stable state}
\newacro{FP}{fictitious play}
\newacro{SFP}{smooth fictitious play}
\newacro{VRFP}{vanishingly regularized fictitious play}
\newacro{DA}{dual averaging}
\newacro{FTRL}{``follow the regularized leader''}
\newacro{RFP}{regularized fictitious play}
\newacro{ODE}{ordinary differential equation}
\newacro{NE}{Nash equilibrium}
\newacroplural{NE}[NE]{Nash equilibria}
\newacro{RNE}{regularized Nash equilibrium}
\newacroplural{RNE}[RNE]{regularized Nash equilibria}
\newacro{VI}{variational inequality}
\newacroplural{VI}[VIs]{variational inequalities}
\newacro{MVI}{Minty variational inequality}
\newacroplural{MVI}[MVIs]{Minty variational inequalities}
\newacro{SVI}{Stampacchia variational inequality}
\newacroplural{SVI}[SVIs]{Stampacchia variational inequalities}

\begin{abstract}
%
%
We examine the long-run behavior of a wide range of dynamics for learning in nonatomic games, in both discrete and continuous time.
The class of dynamics under consideration includes \acl{FP} and its regularized variants, the \acl{BRD} (again, possibly regularized), as well as the dynamics of \acl{DA}\,/\,\acl{FTRL} (which themselves include as special cases the replicator dynamics and Friedman's projection dynamics).
Our analysis concerns both the actual trajectory of play and its time-average, and we cover potential and monotone games, as well as games with an \acl{ESS} (global or otherwise).
We focus exclusively on games with finite action spaces;
nonatomic games with continuous action spaces are treated in detail in Part II.
\end{abstract}

\allowdisplaybreaks		
\acresetall		
\maketitle


\section{Introduction}
\label{sec:introduction}

The study of evolutionary dynamics in population games has been the mainstay of evolutionary game theory ever since the inception of the field in the mid-1970's.
Such dynamics are typically derived from biological models
or economic microfoundations that express the net growth rate of a type (or strategy) within a population
as a function of the relative frequency of all other types.
Thus, owing to the flexibility of this general framework, these considerations have generated an immense body of literature;
for an introduction to the topic, we refer the reader to the masterful treatments of \citet{HS98,Wei95}, and the more recent, panoramic survey of \citet{San10,San15}.

Our paper takes a complementary approach and aims to examine the long-run behavior of \emph{learning procedures} in nonatomic games.
While the topic of learning in games with a \emph{finite} number of players is fairly well-studied \textendash\ see \eg the classic works of \citet{FL98} and \citet{CBL06} \textendash\ the same cannot be said for learning in population games.
On that account, our goal in this paper is to provide a synthetic treatment for a wide range of dynamics for learning in nonatomic games, unifying along the way a number of existing results in the literature.

We pursue this goal in two parts.
In the current paper, we focus exclusively on games with shared, \emph{finite} action spaces;
games with \emph{continuous} (and possibly player-dependent) action spaces are significantly more complicated, so their study is deferred to the companion paper \citep{HLMS21b}.

\para{Overview of the paper}
To specify the general framework under consideration one has to:
\begin{enumerate}
[left=\parindent,label={\itshape\alph*\upshape)}]
\item
Describe the family of games to be studied:
These are anonymous nonatomic games and two concrete classes thereof, potential and monotone games;
the relevant details are provided in \cref{sec:prelims}.
\item
Identify the solution concepts involved:
We will mainly focus on \aclp{NE} and variants of \aclp{ESS}, expressed throughout as solutions of an associated \acl{VI};
this is also described in \cref{sec:prelims}.
\item
Define the dynamics under study:
This is the content of \cref{sec:dynamics}, where we begin by recalling the standard \acl{FP} procedure of \citet{Bro51} and \citet{Rob51}, its regularized variants \citep{HS02}, and its continuous-time analogue, the \acl{BRD} of \citet{GM91}.
Subsequently, we introduce a general class of dynamics known as \acli{DA} \textendash\ or \acli{FTRL} \textendash\ and which contains as special cases the replicator dynamics of \citet{TJ78} and the projection dynamics of \citet{Fri91}.
\end{enumerate}

The dynamics we examine evolve in both discrete and continuous time:
the continuous-time analysis is presented in \cref{sec:cont} and comprises the template for the discrete-time analysis that follows in \cref{sec:disc}.
In both cases, our analysis builds on previous work and results by \citet{MS96,San01,HS03,HS02,HS09,HSV09} and \citet{MS16,MerSan18}.
In this context, some of the new results that we present in our paper can be summarized as follows:
\begin{enumerate}
[left=\parindent,label={\itshape\alph*\upshape)}]
\item
We establish a precise link between (regularized) \acl{FP} and \acl{DA}:
informally, \acl{FP} best-responds to the average population state whereas \acl{DA} best-responds to the average population payoff profile.
\item
We determine the equilibrium convergence properties of \acl{VRFP} in continuous time.
\item
We provide a range of conditions for the convergence of \acl{DA} trajectories to \aclp{ESS}, in both discrete and continuous time (and possibly with a variable learning rate).
\end{enumerate}
Finally, in \cref{sec:extensions}, we discuss several extensions of these results beyond single-population nonatomic games, and we provide a preview of the issues that arise in nonatomic games with continuous action spaces \textendash\ the topic of interest of the follow-up paper \citep{HLMS21b}.

\section{Preliminaries}
\label{sec:prelims}

\subsection{Nonatomic games}

A (single-population) \emph{nonatomic game} is defined by the following primitives:

\begin{itemize}
[left=\parindent]
\addtolength{\itemsep}{\smallskipamount}

\item
A \emph{population of players}.
This is modeled by the unit interval $\players = [0,1]$ endowed with the Lebesgue measure $\leb$.
The nonatomic aspect means that the behavior of a set of players with zero measure has no impact on the game (see below).

\item
A set of \emph{pure strategies} $\pures = \{1,\dotsc,\nPures\}$.
Following \citet{Sch73}, the players' \emph{strategic behavior} is defined as a measurable function $\chi \from \players \to \pures$, with $\chi(\play) \in \pures$ denoting the pure strategy of player $\play$.
Given a pure strategy $\pure\in\pures$, the inverse image $\chi^{-1}(\pure) = \setdef{\play\in\players}{\chi(\play) = \pure}$ will be called the set of ``$\pure$-strategists'', and the corresponding pushforward measure $\strat \defeq \chi \sharp \leb \equiv \leb\circ\chi^{-1}$ on $\pures$ will be called the \emph{state} of the population.
This measure will be the main variable of interest;
for convenience, we will view $\strat$ as an element of the simplex $\strats \defeq \simplex(\pures)$ and we will write $\strat_{\pure} = \leb(\chi^{-1}(\pure))$ for the \emph{relative frequency} of $\alpha$-strategists in the population.

\item
A family of continuous \emph{payoff functions} $\pay_{\pure}\from\strats\to\R$, $\pure\in\pures$.
Specifically, $\pay_{\pure}(\strat)$ denotes the payoff to $\pure$-strategists when the state of the population is $\strat\in\strats$.
Collectively, we will also write $\payv(\strat) = (\pay_{\pure}(\strat))_{\pure\in\pures}$ for the population's \emph{payoff vector} at state $\strat$.
\end{itemize}

The above definition specifies a \emph{population game} $\game \equiv \gamefull$, \cf \citet{San10} and references therein.
In this definition, players are \emph{anonymous} in the sense that their payoffs factor through the state of the population and do not depend on specific players choosing specific strategies.
For more intricate classes of nonatomic games (possibly involving idiosyncratic components), see \citet{Sch73} and \citet{MC84}.

\subsection{Equilibria, stability, and \aclp{VI}}

In population games, an equilibrium is defined as a population state $\eq\in\strats$ in which almost all players are satisfied with their choice of strategy.
Formally, $\eq\in\strats$ is an \emph{equilibrium} of $\game$ if
\begin{equation}
\label{eq:Nash}
\tag{Eq}
\pay_{\pure}(\eq)
	\geq \pay_{\purealt}(\eq)
	\quad
	\text{for all $\pure\in\supp(\eq)$, $\purealt\in\pures$},
\end{equation}
where $\supp(\eq)$ denotes the support of $\eq$.
Equivalently, this requirement for $\eq\in\strats$ can be stated as a \acdef{SVI} of the form%
\begin{equation}
\label{eq:SVI}
\tag{SVI}
\braket{\payv(\eq)}{\strat - \eq}
	\leq 0
	\quad 
	\text{for all $\strat \in \strats$}.
\end{equation}
Equilibria can also be seen as fixed points of the continuous endomorphism $\strat \mapsto \Eucl_{\strats}(\strat + \payv(\strat))$,
where $\Eucl_{\strats}$ denotes the Euclidean projector on $\strats$.
By standard fixed point arguments, it follows that the game's set of equilibria is nonempty;
we will denote this set as $\Eq(\game)$.

\begin{remark*}
In operator theory and optimization, the direction of \eqref{eq:SVI} is typically reversed because optimization problems are usually formulated as \emph{minimization} problems.
The utility maximization viewpoint is more common in the literature on population games, so we will maintain this sign convention throughout.
\end{remark*}

An alternative characterization of equilibria is by means of the \emph{best-response correspondence} $\BR\from\strats\too\strats$, defined here as
\begin{align}
\label{eq:BR}
\BR(\strat)
	&\defeq \argmax_{\stratalt\in\strats} \braket{\payv(\strat)}{\stratalt},
\end{align}
We then have the following string of equivalences:
\begin{equation}
\label{eq:characterizations}
\eq
	\in \BR(\eq)
	\iff
\supp(\eq)
	\subseteq \BR(\eq)
	\iff
\eq
	\in \Eq(\game)
\end{equation}
\ie the equilibria of the game correspond to the fixed points of $\BR$. One can easily show that, under our assumptions, this correspondence has a closed graph and convex compact non-empty values.
By Kakutani's fixed point theorem, $\Eq(\game)$ \textendash\ that is, the set of fixed points of $\BR$ \textendash\ is nonempty and compact.

We will also consider the associated
\acdef{MVI} for $\eq\in\strats$,
namely
\begin{equation}
\label{eq:MVI}
\tag{MVI}
\braket{\payv(\strat)}{\strat - \eq}
	\leq 0
	\quad
	\text{for all $\strat \in \strats$}.
\end{equation}
For concreteness, we will write $\SVI(\game)$ and $\MVI(\game)$ for the set of solutions of \eqref{eq:SVI} and \eqref{eq:MVI} respectively.
In terms of structural properties,
$\MVI(\game)$ is convex
and,
under our continuity assumption for $\payv$, it is straightforward to verify that $\MVI(\game) \subseteq \SVI(\game) \equiv \Eq(\game)$.
In this regard, \eqref{eq:MVI} represents an ``equilibrium refinement'' criterion for $\game$ (see also \cref{sec:classes} below).

In more detail, \eqref{eq:MVI} is intimately related to the seminal notion of \emph{evolutionary stability} that was introduced by \citet{MSP73} and formed the starting point of evolutionary game theory.
To make this link precise, recall that a population state $\eq\in\strats$ is an \acli{ESS} of $\game$ \citep{MSP73,HS98,San10}, if
\begin{equation}
\label{eq:ESS}
\tag{ESS}
\braket{\payv(\delta\strat + (1-\delta)\eq)}{\strat - \eq}
	< 0
\end{equation}
for all sufficiently small $\delta>0$ and all  states $\strat\neq\eq$.
Analogously, $\eq\in\strats$ is a \acdef{NSS} of $\game$ if \eqref{eq:ESS} holds as a weak inequality, \ie
\begin{equation}
\label{eq:NSS}
\tag{NSS}
\braket{\payv(\delta\strat + (1-\delta)\eq)}{\strat - \eq}
	\leq 0
\end{equation}
for all sufficiently small $\delta>0$ and all $\strat\in\strats$.
Finally, $\eq$ is said to be a \acli{GESS} (resp.~\acli{GNSS}) of $\game$ if \eqref{eq:ESS} [resp.~\eqref{eq:NSS}] holds for all $\delta\in(0,1]$.
In obvious notation, we will write $\ESS(\game)$, $\NSS(\game)$, $\GESS(\game)$ and $\GNSS(\game)$ for the respective set of states of $\game$ (by definition $\GESS(\game)$ consists of at most a single point).

\para{Characterizations and implications}
A concise characterization of these notions of stability was derived by \citet{Tay79} and \citet{HSS79}, who showed that $\eq\in\ESS(\game)$ if and only if there exists a neighborhood $\nhd$ of $\eq$ in $\strats$ such that
\begin{equation}
\label{eq:ESS-VI}
\braket{\payv(\strat)}{\strat - \eq}
	< 0
	\quad
	\text{for all $\strat \in \nhd\setminus\{\eq\}$},
\end{equation}
and, furthermore, if we can take $\nhd=\strats$ above, then $\GESS(\game) = \{\eq\}$.
For the corresponding neutral versions, we have $\eq\in\NSS(\game)$ if and only if \eqref{eq:ESS-VI} holds on $\nhd$ as a weak inequality, and $\eq\in\GNSS(\game)$ if and only if we can take $\nhd = \strats$ in this case.
This characterization of $\GNSS(\game)$ corresponds precisely to the solution set of \eqref{eq:MVI};
in particular, we have the following string of inclusions:
\begin{equation}
\label{eq:taxonomy}
\begin{aligned}
\begin{tikzcd}[column sep=tiny,row sep=small]
\GESS
	\arrow[draw=none]{r}[sloped,auto=false]{\subseteq}
	\arrow[draw=none]{d}[sloped,auto=false]{\subseteq}
	&\GNSS
	\arrow[draw=none]{r}[sloped,auto=false]{\equiv}
	\arrow[draw=none]{d}[sloped,auto=false]{\subseteq}
	&\MVI
	\arrow[draw=none]{d}[sloped,auto=false]{\subseteq}
	&
	\\
\ESS
	\arrow[draw=none]{r}[sloped,auto=false]{\subseteq}
	&\NSS
	\arrow[draw=none]{r}[sloped,auto=false]{\subseteq}
	&\Eq
	\arrow[draw=none]{r}[sloped,auto=false]{\equiv}
	&\SVI
\end{tikzcd}
\end{aligned}
\end{equation}
In general, the above inclusions are all one-way;
in the next section, we discuss two important cases where some of them become equivalences.
\acused{ESS}
\acused{NSS}
\acused{GESS}
\acused{GNSS}

\subsection{Classes of games}
\label{sec:classes}

The following classes of games will play an important role in the sequel:
\begin{enumerate}
\addtolength{\itemsep}{\medskipamount}

\item
\emph{Potential games:}
these are games that admit a \emph{potential function}, \ie a function $\pot\from\strats\to\R$ such that $\payv = \nabla\pot$ on $\strats$, or, more generally,
\begin{equation}
\label{eq:pot}
\tag{Pot}
\pot'(\strat;\stratalt-\strat)
	= \braket{\payv(\strat)}{\stratalt-\strat}
	\quad
	\text{for all $\strat,\stratalt \in \strats$},
\end{equation}
where $\pot'(\strat;\tvec) \defeq \lim_{t\to0^{+}} \bracks{\pot(\strat+t\tvec) - \pot(\strat)} / t$ denotes the one-sided directional derivative of $\pot$ at $\strat$ along $\tvec$.
If $\game$ is a potential game, any local maximum of $\pot$ is an equilibrium of $\game$, and any \emph{strict} local maximum of $\pot$ is an \ac{ESS} of $\game$.
For a detailed treatment of potential games in population games extending the initial study of \citet{MS96}, see \citet{San01,San10}.

\item
\emph{Monotone games:}
	these are games that satisfy the monotonicity condition
\begin{equation}
\label{eq:monotone}
\tag{Mon}
\braket{\payv(\stratalt) - \payv(\strat)}{\stratalt - \strat}
	\leq 0
	\quad
	\text{for all $\strat,\stratalt\in\strats$}.
\end{equation}
This condition implies that every equilibrium of $\game$ is neutrally stable, so we get the following equivalences:
\begin{equation}
\MVI(\game)
	\equiv \GNSS(\game)
	= \NSS(\game)
	= \Eq(\game)
	\equiv \SVI(\game)
\end{equation}
In this class, the existence of equilibria relies only on the minmax theorem \citep{Min67}.

Going further, a game is called \emph{strictly monotone} if \eqref{eq:monotone} holds as a strict inequality for all $\stratalt\neq\strat$.
In this case, we have the stronger equivalence
\begin{equation}
\GESS(\game)
	= \ESS(\game)
	= \Eq(\game),
\end{equation}
\ie all inclusions in \eqref{eq:taxonomy} become two-way.
Moreover, $\Eq(\game)$ is a singleton which is also the unique \ac{GESS} of the game.
\end{enumerate}

The intersection of potential and monotone (resp. \emph{strictly} monotone) games occurs when $\game$ admits a \emph{concave} (resp. \emph{strictly} concave) potential.
The most important example of such games is the class of nonatomic congestion games that arise in routing and transportation theory, \cf~\citet{Daf80}.

\subsection{Regularization}

In the sequel, we will often consider ``regularized'' best responses that are single-valued.
These are defined as follows:

\begin{definition}
\label{def:reg}
A \emph{regularizer on $\strats$} is a function $\hreg\from\R^{\pures}\to\R\cup\{\infty\}$ such that:
\begin{enumerate}
\item
$\hreg$ is supported on $\strats$, \ie $\dom\hreg = \setdef{\strat\in\R^{\pures}}{\hreg(\strat)<\infty} = \strats$.
\item
$\hreg$ is strictly convex and continuous on $\strats$.
\end{enumerate}
Given a regularizer on $\strats$, we further define:
\begin{subequations}
\begin{enumerate}
[\itshape a\upshape)]
\item
The associated \emph{choice map} $\choice{\hreg}\from\R^{\pures}\to\strats$ given by
\begin{alignat}{2}
\label{eq:choice}
\choice{\hreg}(\payvec)
	&\defeq \argmax_{\stratalt\in\strats} \{ \braket{\payvec}{\stratalt} - \hreg(\stratalt) \}.
\shortintertext{%
\item
The \emph{regularized best response} to a population state $\strat\in\strats$:
}
\label{eq:BR-reg}
\rBR{\hreg}(\strat)
	&\defeq \argmax_{\stratalt\in\strats} \{ \braket{\payv(\strat)}{\stratalt} - \hreg(\stratalt) \}
	= \choice{\hreg}(\payv(\strat)).
\end{alignat}
\end{enumerate}
\end{subequations}
\end{definition}

\begin{remark}
By our assumptions for $\hreg$ (continuity and strict convexity), $\choice{\hreg}$ and $\rBR{\hreg}$ are both well-defined and single-valued as correspondences.
Moreover, by the strict convexity of $\hreg$, it follows that the convex conjugate $\hconj(\score) = \max_{\strat\in\strats}\{ \braket{\score}{\strat} - \hreg(\strat) \}$ of $\hreg$ is differentiable and satisfies
\begin{equation}
\label{eq:dhconj}
\choice{\hreg}(\payvec)
	= \nabla\hconj(\payvec)
	\quad
	\text{for all $\dvec\in\R^{\pures}$}.
\end{equation}
\end{remark}

\begin{remark}
By construction, $\rBR{\hreg}{}$ (approximately) best-responds to strategies, while $\choice{\hreg}$ best-responds to payoff vectors.
For this reason, we will sometimes refer to $\choice{\hreg}$ as a ``payoff-based'' regularized best response 
(as opposed to ``strategy-based'' regularized best response for $\rBR{\hreg}{}$).
\end{remark}

Moving forward, for part of our analysis, we will also require one (or both) of the regularity conditions below:
\begin{enumerate}
\addtolength{\itemsep}{\smallskipamount}

\item
$\hreg$ is \emph{$\hstr$-strongly convex} on $\strats$, \ie
\begin{equation}
\label{eq:strong}
\tag{StrCvx}
\hreg(\coef\strat + (1-\coef)\stratalt)
	\leq \coef \hreg(\strat)
		+ (1-\coef) \hreg(\stratalt)
		- \frac{\hstr}{2} \coef (1-\coef) \norm{\stratalt - \strat}^{2}
\end{equation}
for all $\strat,\stratalt\in\strats$ and all $\coef\in[0,1]$.

\item
The subdifferential $\subd\hreg$ of $\hreg$ admits a \emph{continuous selection};
specifically, writing $\proxdom \defeq \dom\subd\hreg = \setdef{\strat}{\subd\hreg(\strat)\neq\varnothing}$ for the ``prox-domain'' of $\hreg$, we posit that there exists a continuous map $\subsel\hreg\from\proxdom\to\R^{\pures}$ such that
\begin{equation}
\label{eq:selection}
\tag{Diff}
\subsel\hreg(\strat)
	\in \subd\hreg(\strat)
	\quad
	\text{for all $\strat\in\proxdom$}.
\end{equation}
\end{enumerate}

By replacing best responses with regularized best responses in the definition of equilibria, we also define the notion of a \emph{regularized equilibrium} as follows:

\begin{definition}
\label{def:Nash-reg}
Given a regularizer $\hreg$ on $\strats$ and a regularization weight $\weight>0$, we define an \emph{$\weight$-regularized equilibrium} of $\game$ as a profile $\eq\in\strats$ such that $\eq\in\rBR{\weight\hreg}(\eq)$.
The set of $\weight$-regularized equilibria of $\game$ will be denoted as $\rEq{\weight\hreg}(\game)$.
\end{definition}

Under \eqref{eq:selection}, $\eq\in\strats$ is an $\weight$-regularized equilibrium if and only if
\begin{equation}
\label{eq:Nash-eps}
\braket{\payv(\eq) - \weight\subsel\hreg(\eq)}{ \eq -\strat }
	\geq 0
	\quad
	\text{for all $\strat\in\strats$}.
\end{equation}
In this regard, $\weight$-regularized equilibria can be seen as \aclp{NE} of the ``$\weight$-regularized'' game $\game_{\weight} \equiv \game_{\weight}(\pures,\payv_{\weight})$ with payoff profile $\payv_{\weight}(\strat) \defeq \payv(\strat) -  \weight\subsel\hreg(\strat)$;
more concisely, $\rEq{\weight}(\game) = \Eq(\game_{\weight})$.

Equilibria and regularized equilibria are related by the next hemicontinuity property, itself a consequence of the maximum theorem of \citet{Ber97}. 

\begin{lemma}
\label{lem:cont}
Let $\eq_{\weight} \in \rEq{\weight}(\game)$, $\weight>0$, be a family of regularized equilibria of $\game$.
Then, in the limit $\weight\to0$, any accumulation point of $\eq_{\weight}$ is an equilibrium of $\game$.
\end{lemma}

We close this section with a concrete illustration of the above concepts:

\begin{example}
[Logit choice]
\label{ex:logit}
An important example of regularization corresponds to the entropic regularizer $\hreg(\strat) = \sum_{\pure\in\pures} \strat_{\pure}\log\strat_{\pure}$ which, by a standard calculation, leads to the choice map
\begin{equation}
\label{eq:logit}
\logit(\payvec)
	= \frac{(e^{\payvec_{1}},\dotsc,e^{\payvec_{\nPures}})}{e^{\payvec_{1}} + \dotsm + e^{\payvec_{\nPures}}}.
\end{equation}
This map is commonly referred to as the \emph{logit choice map} \citep{McF74a}, and the corresponding fixed points $\eq = \logit(\payv(\eq)/\weight)$ are known as \emph{logit equilibria};
for a detailed discussion, see \citet{vD87,MP95,San10}, and references therein.
It is also well known that the entropic regularizer is $1$-strongly convex relative to the $L_{1}$-norm, so $\logit$ is $1$-Lipschitz continuous, \cf \citet{SS11}.
\endenv
\end{example}


\section{Dynamics}
\label{sec:dynamics}

In this section, we introduce a wide range of dynamics for learning in nonatomic games, in both continuous and discrete time.
The unifying principle of the dynamics under study is that agents play a best response \textendash\ regularized or otherwise \textendash\ to some version of the history of play:
the empirical population frequencies, or the players' average gains.
For historical reasons, we begin with the discrete-time framework in \cref{sec:dyn-disc}, and we present the corresponding continuous-time dynamics in \cref{sec:dyn-cont}.

\subsection{Discrete-time processes}
\label{sec:dyn-disc}

In the discrete-time case, we will consider the evolution of the population state $\curr \in \strats$ over a series of epochs $\run = \running$ (corresponding for example to generations in an evolutionary context).
In turn, this population state is determined by a recursive update rule \textendash\ or \emph{algorithm} \textendash\ that defines the dynamics in question.
Depending on the manner of ``best-responding'' and the signal that the agents are best-responding to, we have the following dynamics.

\para{\Acl{FP}}
Dating back to \citet{Bro51} and \citet{Rob51}, \acdef{FP} posits that agents best-respond to the \emph{empirical frequency} (or ``\emph{time average}'') $\curr[\avg] = (1/\run) \sum_{\runalt=\start}^{\run} \iter$ of their population state.
Concretely, this means that the population state evolves according to the process
\begin{equation}
\label{eq:FP}
\tag{\acs{FP}}
\next
	\in \BR(\curr[\avg])
	= \argmax\nolimits_{\strat\in\strats} \braket{\payv(\curr[\avg])}{\strat}.
\end{equation}
In terms of empirical frequencies, we have the recursive dynamics
\begin{equation}
\label{eq:FP-avg}
\next[\avg] - \curr[\avg]
	\in \frac{1}{\run+1} \bracks{\BR(\curr[\avg]) - \curr[\avg]}.
\end{equation}

In contrast to the sequence of actual population states $\curr\in\strats$, $\run=\running$, the stationary points of the time-averaged process \eqref{eq:FP-avg} are straightforward to characterize.
Indeed, any fixed point $\eq$ of \eqref{eq:FP-avg} satisfies
\begin{equation}
\label{eq:FP-rest}
0
	\in \BR(\eq) - \eq,
\end{equation}
\ie $\eq$ is stationary under \eqref{eq:FP-avg} if and only if it is an equilibrium of $\game$.

\para{\Acl{RFP}}
A variant of the original \acl{FP} algorithm is the ``regularized'' version in which best responses are replaced by \emph{regularized} best responses.
In particular, given an underlying regularizer $\hreg$ on $\strats$ and a corresponding regularization weight $\weight>0$, the \acdef{RFP} process is defined as
\begin{equation}
\label{eq:RFP}
\tag{\acs{RFP}}
\next
	= \rBR{\weight\hreg}(\curr[\avg])
	= \argmax\nolimits_{\strat\in\strats} V_{\weight\hreg}(\curr[\avg], {\strat}) 
\end{equation}
with 
\begin{equation}
\label{def:V}
V_{\weight\hreg}(\stratalt,\strat)
	= \braket{\payv(\stratalt)}{\strat} - \weight \hreg(\strat)
\end{equation}
Then, as before, in terms of empirical frequencies, we have
\begin{equation}
\label{eq:RFP-avg}
\next[\avg] - \curr[\avg]
	\in \frac{1}{\run+1} \bracks{\rBR{\weight\hreg}(\curr[\avg]) - \curr[\avg]}.
\end{equation}

As in the case of \eqref{eq:FP}, the stationary points of the time-averaged dynamics \eqref{eq:RFP} are given by the fixed point equation
\begin{equation}
\label{eq:RFP-rest}
0
	\in \rBR{\weight\hreg}(\eq) - \eq,
\end{equation}
\ie they are precisely the $\weight$-regularized equilibria of $\game$ (\cf \cref{def:Nash-reg}).

\begin{remark*}
In the literature, this process is sometimes referred to as ``smooth'' or ``perturbed'' \acl{FP}, \cf \citet{FL99,HS02} and references therein.
This difference in terminology is owed to a different set of assumptions for $\hreg$, which gives rise to interior-valued choice maps $\choice{}\from\R^{\pures}\to\strats$.
\end{remark*}

\para{\Acl{VRFP}}

When the regularization weight in \eqref{eq:RFP} decreases over time, we obtain the \acdef{VRFP} dynamics of \citet{BF13}, viz.
\begin{equation}
\label{eq:VRFP}
\tag{VRFP}
\next
	= \rBR{\curr[\weight]\hreg}(\curr[\avg])
	= \argmax\nolimits_{\strat\in\strats} \ V_{\curr[\weight]\hreg}( \curr[\avg], {\strat}) 
\end{equation}
with $\curr[\weight] > 0$ for all $\run$ and $\curr[\weight]\to0$ as $\run\to\infty$.
 In terms of empirical frequencies, we then have
\begin{equation}
\label{eq:VRFP-avg}
\next[\avg] - \curr[\avg]
	\in \frac{1}{\run+1} \bracks{\rBR{\curr[\weight]\hreg}(\curr[\avg]) - \curr[\avg]}.
\end{equation}
Because this process is non-autonomous ($\curr[\weight]$ depends explicitly on $\run$), it is no longer meaningful to discuss its rest points.

\para{\Acl{DA}}
Dually to the above, instead of best-responding to the aggregate history of \emph{play} (perfectly or approximately), we can also consider the case where agents play a ''regularized best response to their aggregate \emph{payoff}''.
In our context, this gives rise to the \acli{DA} dynamics
\begin{align}
\label{eq:DA-agg}
\next
	&= \choice{\hreg}\parens*{\curr[\learn] \sum_{\runalt=\start}^{\run} \payv(\iter)},
\shortintertext{with $\choice{\hreg}$ defined in \eqref{eq:choice}.
Then, in iterative form, we get the update rule}
\label{eq:DA}
\tag{\acs{DA}}
\begin{split}
\curr[\aggpay]
	&= \prev[\aggpay]
		+ \payv(\curr)
	\\
\next
	&= \choice{\hreg}(\curr[\learn]\curr[\aggpay])
\end{split}
\end{align}
where $\curr[\learn] > 0$ is a ``learning rate'' parameter in the spirit of \citet{Nes09}.
In the recursive formulation \eqref{eq:DA} of the dynamics, the initialization $\beforeinit[\aggpay] = 0$ is the most common one but, otherwise, it can be arbitrary.

\begin{remark*}
In the literature on online learning and online optimization, the process \eqref{eq:DA} is known as \acdef{FTRL} \citep{SS11,SSS06}.
The terminology ``dual averaging'' is due to \citet{Nes09} and is more common in the theory of \aclp{VI} and offline convex optimization.
We adopt the latter to highlight the link of our work to \aclp{VI}.
\end{remark*}

\para{Links between the dynamics: the case of random matching}

To illustrate the relations between the various learning processes discussed above, it will be instructive to consider nonatomic games generated by random matching \citep{Wei95,San10}.
In such games, players are randomly matched to play a symmetric two-player game with payoff matrix $A\in\R^{\pures\times\pures}$, so the population's mean payoff profile is $\payv(\strat) = M\strat$ for all $\strat\in\strats$.
As a result, the dynamics \eqref{eq:DA} may be rewritten as
\begin{equation}
\next
	= \choice{\hreg}\parens*{\curr[\learn] \sum_{\runalt=\start}^{\run} \payv(\iter)}
	= \choice{\hreg}\parens*{\run\curr[\learn] \payv(\curr[\avg])}
	= \choice{\hreg/(\run\curr[\learn])}(\payv(\curr[\avg])),
\end{equation}
which immediately allows us to recover the variants of \acl{FP} as follows:
\begin{enumerate}
\item
For \eqref{eq:VRFP}:
	$\curr[\learn] = 1/(\run\curr[\weight])$ for a sequence of weights $\curr[\weight]>0$, $\lim_{\run\to\infty} \curr[\weight] = 0$.
\item
For \eqref{eq:RFP}:
	$\curr[\learn] = 1/(\run\weight)$ for some fixed $\weight>0$.
\item
For \eqref{eq:FP}:
	this corresponds to the limiting cases ``$\run\curr[\learn] = \infty$'' or ``$\hreg = 0$''.
[Of course, these parameter choices are not formally allowed in \eqref{eq:DA}, so this limit is informal.]
\end{enumerate}

\begin{remark*}
The above is meaningful only when $\payv(\strat)$ is linear in $\strat$;
however, the relation between $\weight$ and $\learn$ will be seen to underlie a large part of the sequel.
\end{remark*}

\subsection{Continuous-time dynamics}
\label{sec:dyn-cont}

We now proceed to define the corresponding continuous-time dynamics for each of the processes described above.

\para{\Acl{BRD}}

The autonomous formulation \eqref{eq:FP-avg} of \acl{FP} can be seen as an Euler discretization of the \acdef{BRD} of \citet{GM91}, namely
\begin{equation}
\label{eq:BRD}
\tag{\acs{BRD}}
\dot\mean_{\time}
	\in \BR(\mean_{\time}) - \mean_{\time}.
\end{equation}
In the above, $\mean_{\time} \in \strats$ is the continuous-time analogue of the empirical mean process $\curr[\avg]$ in discrete time;
we use the notation $\mean$ instead of $\strat$ to highlight this link.
Clearly, the rest points of \eqref{eq:BRD} are again described by the fixed point equation \eqref{eq:FP-rest}, \ie $\eq[\mean]$ is stationary under \eqref{eq:BRD} if and only if it is an equilibrium of $\game$.

\para{\Acl{RBRD}}

Working as above, the \acdef{RBRD} are defined as
\begin{equation}
\label{eq:RBRD}
\tag{\acs{RBRD}}
\dot\mean_{\time}
	= \rBR{\weight\hreg}(\mean_{\time}) - \mean_{\time}.
\end{equation}
As in the case of \eqref{eq:RFP}, the rest points of \eqref{eq:RBRD} are characterized by the fixed point equation \eqref{eq:RFP-rest}, \ie $\eq[\mean]$ is stationary under \eqref{eq:RBRD} if and only if it is an $\weight$-regularized equilibrium of $\game$.

\begin{example}
[The logit dynamics]
\label{ex:RBRD-logit}
Building on the logit choice map discussed in \cref{sec:prelims} (\cf \cref{ex:logit}), the associated \acl{RBRD} are known as \emph{logit dynamics}, and are given by
\begin{equation}
\label{eq:LD}
\dot\mean_{\time}
	= \logit\parens*{\payv(\mean_{\time})/ \weight} - \mean_{\time}.
\end{equation}
\end{example}

\para{ Vanishing regularized best reply dynamics}

If we allow the regularization weight $\weight$ to vary in \eqref{eq:RBRD}, we obtain the non-autonomous dynamics
\begin{equation}
\label{eq:VBRD}
\tag{\acs{VBRD}}
\dot\mean_{\time}
	= \rBR{\weight_{\time}\hreg}(\mean_{\time}) - \mean_{\time},
\end{equation}
where the ``V'' indicates again that $\weight_{\time}\to 0$ as $\time\to\infty$.

\para{The \acl{DAD}}

Finally, given a regularizer $\hreg$ on $\strats$ as above, the dynamics of \acl{DA} in continuous time can be written as
\begin{equation}
\label{eq:DA-cont}
\tag{\acs{DAD}}
\begin{aligned}
\dot\score_{\time}
	&= \payv(\strat_{\time})
	\\
\strat_{\time}
	&=\choice{\hreg}(\learn_{\time}\score_{\time})
\end{aligned}
\end{equation}
where $\learn_{\time} \geq 0$ is a time-varying learning parameter, which we assume throughout to be non-increasing.
More compactly, the above can also be written as
\begin{subequations}
\begin{align}
\label{eq:DA-diff}
\dot\score_{\time}
	&= \payv(\choice{\hreg}(\learn_{\time}\score_{\time}))
\shortintertext{or}
\label{eq:DA-int}
\strat_{\time}
	&=\choice{\hreg}\parens*{\learn_{\time} \int_{0}^{\time} \payv(\strat_{\timealt}) \dd\timealt},
\end{align}
\end{subequations}
depending on whether we take a differential- or integral-based viewpoint.
Specifically, we note that \eqref{eq:DA-diff} is an autonomous \acl{ODE} evolving in the  dual, payoff space of the game;
by contrast, \eqref{eq:DA-int} is an integral equation stated directly in the primal, strategy space of the game.
We also note that, unlike \eqref{eq:BRD}, \eqref{eq:RBRD} and \eqref{eq:VBRD}, the dynamics \eqref{eq:DA-cont} are stated in terms of the actual population state $\strat_{\time}$ at time $\time$, not its empirical mean $\mean_{\time}$ (which is the driving variable of the previous dynamics).

\begin{example}
If $\learn_{\time} \equiv 1$ and $\hreg$ is the entropic regularizer on $\strats$ (\cf \cref{ex:logit}), the dynamics \eqref{eq:DA-cont} yield the replicator dynamics of \citet{TJ78}, viz.
\begin{equation}
\label{eq:RD}
\tag{RD}
\dot\strat_{\pure,\time}
	= \strat_{\pure,\time}
		\bracks{\pay_{\pure}(\strat_{\time})-\braket{\payv(\strat_{\time})}{\strat_{\time}}}.
\end{equation}
In words, \eqref{eq:RD} indicates that the per capita growth rate of the population of $\pure$-strategists is proportional to the difference between the payoff that they experience and the mean population payoff.
For a detailed presentation of this derivation in different contexts, see \citet{Rus99,Sor09,Sor20,MM10} and references therein.
\endenv
\end{example}

\begin{example}
If $\learn_{\time} \equiv 1$ and $\hreg(\strat) = (1/2) \sum_{\pure\in\pures} \strat_{\pure}^{2}$,
the integral form \eqref{eq:DA-int} of \eqref{eq:DA-cont} becomes
\begin{equation}
\label{eq:PD-int}
\strat_{\time}
	= \Eucl\parens*{\int_{0}^{\time} \payv(\strat_{\timealt}) \dd\timealt}
\end{equation}
where $\Eucl$ denotes the Euclidean projector on $\strats$.
The differential form of \eqref{eq:PD-int} is more delicate to describe because $\strat_{\time}$ may enter and exit different faces of $\strats$ in perpetuity.
However, if we focus on an open interval $\runs\subseteq\R$ over which the support of $\strat_{\time}$ remains constant, it can be shown that \eqref{eq:PD-int} follows the projection dynamics of \citet{Fri91}, viz. for all $\time\in\runs$ we have
\begin{equation}
\label{eq:PD}
\tag{PD}
\dot\strat_{\pure,\time}
	= \begin{cases}
	\pay_{\pure}(\strat) - \frac{1}{\abs{\supp(\strat_{\time})}} \sum_{\purealt\in\supp(\strat_{\time})} \pay_{\purealt}(\strat)
		&\quad
		\text{if $\pure\in\supp(\strat_{\time})$},
		\\
	0
		&\quad
		\text{otherwise}.
	\end{cases}
\end{equation}
For the details of this derivation, see \citet{MS16}.
\endenv
\end{example}

\begin{remark*}
For posterity, we note that both regularizers used in the above examples satisfy \eqref{eq:strong}, \eqref{eq:selection}, and the technical condition \eqref{eq:reciprocity} that we describe later in the paper.
\end{remark*}

\para{Random matching in continuous time}

Following \citet{HSV09}, the various dynamics presented so far can be linked as follows when $\payv$ is linear and  $\mean_{\time} = \frac{1}{\time}\int_{0}^{\time} \strat_{\timealt} \dd\timealt$ for $\time>0$:
\begin{align}
\strat_{\time}
	&= \argmax\limits_{\strat\in\strats}
		\braces*{\learn_{\time} \int_{0}^{\time} \braket{\payv(\strat_{\timealt})}{\strat} \dd\timealt - \hreg(\strat)}
	= \argmax\limits_{\strat\in\strats}
		\braces*{\learn_{\time}\,\time\,  \braket{  \payv(\mean_{\time})}{\strat}  - \hreg(\strat)}
	\notag\\
	&= \rBR{[1/(\time\learn_{\time})]\hreg}\parens{\mean_{\time}}.
\end{align}
It thus follows that the empirical distribution of play $\mean_{\time}$ under \eqref{eq:DA-cont} follows the dynamics
\begin{equation}
\dot\mean_{\time}
	= \frac{\rBR{\weight_{\time}\hreg}(\mean_{\time}) - \mean_{\time}}{\time}
\end{equation}
with $\weight_{\time} = 1/(\time\learn_{\time})$.
Hence, after the change of time $\time \gets \log\time$, we get \eqref{eq:RBRD} with a variable regularization parameter $\weight_{\time} = 1/(\time\learn_{\time})$.

\section{Analysis and results in continuous time}
\label{sec:cont}

We now proceed to present our results for the class of dynamics under study.
We begin with the continuous-time analysis which provides a template for the discrete-time analysis in the next section.
We prove that under some assumptions depending on the dynamical system and the class of games under consideration (potential, monotone or strictly monotone), the dynamics converge to the set of equilibria or or \textendash\ regularized equilibria \textendash\ of the game.%
\footnote{Recall that a family $(\strat_{\time})_{\time\geq0}$ (resp.~a sequence $(\state_{\run})$, $\run=\running$) converges to a set $B$ if the set of accumulation points $ {\textbf A}[x_t] $ (resp. $ \mathbf{A}[x_n] $), as $t \rightarrow \infty$ (resp $n \rightarrow \infty$), satisfies: $ {\textbf A}[x_t] \subset B$ (resp. $ \mathbf{A}[x_n] \subset B$).}

\subsection{\Acl{BRD}}

Our first result concerns the \acl{BRD} \eqref{eq:BRD}.
Albeit slightly more general, the results of this subsection and the next essentially follow \citet{HS02,HS09}.

\begin{theorem}
\label{thm:BRD}
Suppose that one of the following holds:
\begin{enumerate}
[\upshape(\itshape a\upshape)]
\item
$\game$ is potential.
\item
$\game$ is monotone and $\payv$ is $C^1$-smooth.
\end{enumerate}
Then every solution orbit $\mean_{\time}$ of \eqref{eq:BRD} converges to $\Eq(\game)$.
\end{theorem}

\begin{proof}
Both cases rely on establishing a suitable Lyapunov function $\lyap(t)$ for \eqref{eq:BRD}:
\begin{enumerate}
\item
In the potential case, we  simply take $\lyap(\time) = \pot(\mean_{\time})$.
\item
In the monotone case, $\lyap(\time) = \gap(\mean_{\time})$ where $\gap$ denotes the function
\begin{equation}
\gap(\strat)
	:= \max_{\stratalt\in\strats} \braket{\payv(\strat)}{\stratalt-\strat}=\braket{\payv(\strat)}{\BR(x)-\strat}
\end{equation}
Observe that $\gap(\strat) \geq 0$ with equality if and only if $\strat\in\Eq(\game)$.
\end{enumerate}

We now proceed on a case-by-case basis:

\para{Potential case}
From \eqref{eq:pot} we obtain that
\begin{equation}
\label{derivationP}
\ddt \pot(\mean_{\time})
	= \lim_{\delta\to0^{+}} \bracks{\pot(\mean_{\time} + \delta\dot\mean_{\time}) - \pot(\mean_{\time})} / \delta
	= \braket{\payv (\mean_{\time})}{\dot\mean_{\time}}
\end{equation}
Since $\dot\mean_{\time}\in \BR(\mean_{\time}) - \mean_{\time}$, it follows  that $\ddt \pot(\mean_{\time})= \gap (\mean_{\time}) \geq 0$ with equality if and only if $\dot\mean_{\time}=0$. Consequently, from Lyapunov's first method \citep[see \eg][Theorem~2.6.1]{HS98} we deduce that $\mean_{\time}$ converges to the set of rest points of \eqref{eq:BRD}, which is $\Eq(\game)$. 

\para{Monotone case}
Let $\Jac\payv$  denote the Jacobian of $\payv$.
By the envelope theorem we get: 
\begin{equation}
\label{envelope}
\ddt \gap(\mean_{\time})
	= \braket{\payv(\mean_{\time})}{-\dot \mean_{\time}}
		+ \dot \mean_{\time} \Jac\payv (\mean_{\time}) (\BR(\mean_{\time})-\mean_{\time} )
	= -\gap (\mean_{\time})
		+ \dot \mean_{\time} \Jac\payv (\mean_{\time}) \dot \mean_{\time}
\end{equation}
As $\game$ is monotone,  $ (\stratalt-\strat) \Jac\payv(\strat) (\stratalt-\strat) \leq 0$ for all $\strat,\stratalt\in\strats$, so $\dot \mean_{\time} \Jac\payv (\mean_{\time}) \dot \mean_{\time} \leq 0$.
Consequently:
\begin{equation}
\ddt \gap(\mean_{\time})  \leq  -  \gap(\mean_{\time}).
\end{equation}
Hence $\gap(\mean_{\time})$ decreases exponentially to $0$, so $\mean_{\time}$ converges to $\Eq(\game)$.\end{proof}

\subsection{\Acl{RBRD}}

In this section we will assume that the regularizer $\hreg$ satisfies \eqref{eq:selection}.
In a slight abuse of notation, we will skip the dependence on $\hreg$ and write $\payv_{\weight}$ for $\payv_{\weight\hreg} \defeq \payv - \weight \subsel\hreg$, see \eqref{eq:Nash-eps},
with $\subsel\hreg$ defined as in \eqref{eq:selection}.
We also write  $\BR_{\weight}(\strat)$ for $\BR_{\weight \hreg}(\strat)$  and so on.
The reason is that  the regularizer $\hreg$ will  be fixed throughout while the weight $\weight$ will be allowed to vary in the next subsection.
\begin{lemma}
\label{lem:RBRD1}
Under \eqref{eq:selection}, we have $\braket{\payv_{\weight } (\strat)}{\BR_{\weight }(\strat)-\strat} \geq 0$ for all $\strat\in\strats$, and equality holds if and only if $\strat=\BR_{\weight }(\strat)$.
\end{lemma}

\begin{proof}
Let $\phi (z)= \langle \payv(\strat),z \rangle -\weight \hreg(z)$.
Then $\phi$ is strictly concave and so $\partial \phi $ defines a strictly monotone operator.
Thus,  $\strat'=\arg\max_{z\in\strats} \phi(z)$ is a solution of the variational inequality \eqref{eq:MVI} associated to $\partial \phi $: for all $d \in \partial\phi (\strat)$: $\langle d, \strat'-\strat \rangle \geq 0$ with equality if and only if $\strat=\strat'$.
As $\payv_{\weight}(\strat) \in \partial\phi (\strat)$ and $\strat'=\BR_{\weight}(\strat)$ the result follows.
\end{proof}

\begin{lemma}
\label{lem:RBRD2}
Suppose that \eqref{eq:selection} holds.
If $\game$ with payoff field $\payv$ is potential \textpar{resp.~monotone} then the game $\game_{\weight}$ with payoff function $\payv_{\weight}$ is potential \textpar{resp.~strictly monotone}.
\end{lemma}

\begin{proof}
If  $\game$  has a potential $\pot$, a potential function of $\game_{\weight}$ is  $\pot - \weight\hreg$.
If $\game$ is a monotone game, the game $\game_{\weight}$ is strictly monotone because we have $\braket{\subsel\hreg(\strat) - \subsel\hreg(\stratalt)}{\strat-\stratalt} \geq 0$ for every $\strat,\stratalt\in\strats$, with equality if and only if $\strat=\stratalt$.
\end{proof}

We now turn to the asymptotic properties of \acl{RBRD}.

\begin{theorem}
\label{thm:RBRD}
Suppose that \eqref{eq:selection} is satisfied
and assume one of the following holds:
\begin{enumerate}
[\upshape(\itshape a\upshape)]
\item
$\game$ is potential.
\item
$\game$ is monotone and $\payv$ is $C^1$-smooth.
\end{enumerate}
Then every solution orbit $\mean_{\time}$ of \eqref{eq:RBRD} converges to $\rEq{\weight}(\game)$.
\end{theorem}

\begin{remark*}
Observe that when $\game$ is monotone, $\game_{\weight}$ is strictly monotone, so $\rEq{\weight}(\game)$ is a singleton.
\end{remark*}

\begin{proof}
As above, both cases rely on establishing a suitable Lyapunov function $\lyap(t)$ for \eqref{eq:RBRD}:
\begin{enumerate}
\item
In the potential case: $\lyap(\time)=\pot_{\weight} (\mean_{\time}) \defeq \pot(\mean_{\time}) - \weight\hreg(\mean_{\time})$.
\item
In the monotone case: 
\begin{align}
\lyap(\time)
	&= \gap_{\weight}(\mean_{\time})
	= \max_{\stratalt\in\strats} \braket{\payv(\mean_{\time})}{\stratalt-\mean_{\time}}
		- \weight(\hreg(\stratalt) - \hreg(\mean_{\time}))
	\notag\\
	&=  \max_{\stratalt\in\strats}V_{\weight}(\mean_{\time},\stratalt)
		-  V_{\weight}(\mean_{\time}, \mean_{\time}) 
\end{align}
\end{enumerate}

\para{Potential case}
There is $z(\time) \in \partial \hreg(\mean_{\time})$ such that \textendash\ similarly to \eqref{derivationP} \textendash\ we have:
\begin{equation}
\ddt \pot_{\weight}(\mean_{\time})= \lim_{dt\to0^{+}} \bracks{\pot_{\weight}(\mean_{\time}+\dot\mean_{\time}dt) - \pot_{\weight}(\mean_{\time})} / dt= \langle \payv (\mean_{\time}), \dot\mean_{\time}\rangle- \weight \langle z(\time), \dot\mean_{\time}\rangle  
\end{equation}
As $\dot\mean_{\time}= \BR_{\weight }(\mean_{\time}) - \mean_{\time}$, by \cref{lem:choice} in the appendix, we have $\langle z(\time), \dot\mean_{\time}\rangle \leq \langle \subsel\hreg(\mean_{\time}), \dot\mean_{\time}\rangle$.
Consequently
\begin{equation}
\ddt \pot_{\weight}(\mean_{\time}) \geq \langle \payv (\mean_{\time}), \dot\mean_{\time}\rangle- \weight \langle \subsel\hreg(\mean_{\time}), \dot\mean_{\time}\rangle  = \langle \payv_{\weight} (\mean_{\time}), \dot\mean_{\time}\rangle 
\end{equation}
By \cref{lem:RBRD1} and the above, we deduce that $\ddt \pot_{\weight}(\mean_{\time})\geq 0$  with equality if and only if $\dot\mean_{\time}=0$.
Consequently, by a Lyapunov argument, $\mean_{\time}$ converges to $\rEq{\weight}(\game)$.

\para{Monotone case}
Using the envelope theorem as in \eqref{envelope} and, just as above,  there is $z(\time) \in \partial \hreg(\mean_{\time})$ such that: 
\begin{align}
\ddt \gap_{\weight}(\mean_{\time})= -\langle \payv(\mean_{\time})-\weight z(t), \dot \mean_{\time} \rangle + \dot \mean_{\time} \Jac\payv (\mean_{\time}) \dot \mean_{\time}  \leq -\langle \payv_{\weight}(\mean_{\time}), \dot \mean_{\time} \rangle + \dot \mean_{\time} \Jac\payv (\mean_{\time}) \dot \mean_{\time} 
\end{align}
The inequality being a consequence of \cref{lem:choice} as it implies that $\langle z(\time), \dot\mean_{\time}\rangle \leq \langle \subsel\hreg(\mean_{\time}), \dot\mean_{\time}\rangle$.
Since  $\game$ is monotone, $ \dot \mean_{\time} \Jac\payv(\mean_{\time}) \dot \mean_{\time} \leq 0$.
Consequently, as $\dot\mean_{\time}= \BR_{\weight }(\mean_{\time}) - \mean_{\time}$, by \cref{lem:RBRD1}, $\ddt \gap_{\weight}(\mean_{\time})  \leq 0$ with equality if and only if $\dot \mean_{\time}=0$.
Hence  $\mean_{\time}$ converges to $\rEq{\weight}(\game)$.
\end{proof}

\subsection{Vanishing \acl{RBRD}}

A similar technique will be used to obtain convergence to the set of Nash equilibria for vanishing \acl{RBRD}.
The additional difficulty here comes from the fact that the vanishing \acl{RBRD} is not autonomous and so we should prove by hand the Lyapunov convergence argument.
To the best of our knowledge, this is a new result.

\begin{theorem}
\label{thm:VBRD-cont}
Suppose that \eqref{eq:selection} holds and \eqref{eq:VBRD} is run with a smoothly-varying and strictly decreasing regularization weight $\weight_{\time}\to0$.
Assume further that $\hreg$ takes values in $[0,1]$, and one of the following holds:
\begin{enumerate}
[\upshape(\itshape a\upshape)]
\item
$\game$ is potential.
\item
$\game$ is monotone, $\payv$ is $C^1$-smooth, and $\subsel\hreg$ is bounded.
\end{enumerate}
Then every solution orbit $\mean_{\time}$ of \eqref{eq:VBRD} converges to $\rEq{\weight}(\game)$.
\end{theorem}

\begin{remark*}
The assumption that $\hreg$ takes value in $[0,1]$ is only made for notational convenience and does not incur any loss of generality.
\end{remark*}

\begin{proof}
As before, we treat the potential and monotone cases separately.

\para{Potential case}
Let $\pot_{\weight_{\time}} (\mean_{\time}) = \pot (\mean_{\time})- \weight_{\time}\hreg(\mean_{\time})$.
Then, a calculation similar to the above implies that there exists $z(\time) \in \partial \hreg(\mean_{\time})$  such that
\begin{equation}
\ddt \pot_{\weight_{\time}}(\mean_{\time})=  \langle \payv (\mean_{\time})-\weight_{\time} z(\time),  \BR_{\weight_{\time} }(\mean_{\time}) - \mean_{\time}\rangle-\dot\weight_{\time} \hreg(\mean_{\time}) \geq \langle \payv_{\weight_{\time}} (\mean_{\time}),  \BR_{\weight_{\time} }(\mean_{\time}) - \mean_{\time}\rangle-\dot\weight_{\time} \hreg(\mean_{\time})
\end{equation}
The inequality being a consequence of \cref{lem:choice} as it implies that $\langle z(\time), \BR_{\weight_{\time} }(\mean_{\time}) - \mean_{\time}\rangle \leq \langle \subsel\hreg(\mean_{\time}), \BR_{\weight_{\time} }(\mean_{\time}) - \mean_{\time}\rangle$.
By \cref{lem:RBRD1},  $\langle \payv_{\weight_{\time}} (x), \BR_{\weight_{\time} }(x)-x \rangle \geq 0$.
Since $-\dot\weight_{\time} \hreg(\mean_{\time}) \geq 0$ we deduce that $\ddt \pot_{\weight_{\time}}(\mean_{\time}) \geq 0$ and  consequently $\pot_{\weight_{\time}}(\mean_{\time})$ converges.
As $\weight_{\time}$ goes to zero and $\hreg(\cdot)$ is bounded, $\pot(\mean_{\time})$ converges as well.
We will deduce from this that $\mean_{\time}$ converges to $\Eq(\game)$.
By contradiction, if a limit point $\mean_{0}=\lim_k \mean_{\time_k}$ is not in $\Eq(\game)$ then:
\begin{equation}
\gap (\mean_{0})
	= \max_{\stratalt\in\strats} \braket{\payv(\mean_{0})}{\stratalt-\mean_{0}}
	= \braket{\payv(\mean_{0})}{\BR(\mean_{0}) - \mean_{0}}
	= \alpha
	>0.
\end{equation} 
Recall that, by definition, for all $t$ and $\mean$:
\begin{align} 
\label{gap-equality}
\gap_{\weight_{\time}}(\mean)
	&= \max_{\stratalt\in\strats} \braket{\payv(\mean)}{\stratalt-\mean}-\weight_{\time} (\hreg(\stratalt)-\hreg(\mean))
	\notag\\
	&=\braket{\payv(\mean)}{\BR_{\weight_{\time} }(\mean)-\mean}-\weight_{\time} (\hreg(\BR_{\weight_{\time} }(\mean))-\hreg(\mean))
	\notag\\
	&= \max_{\stratalt\in\strats}V_{\weight_{\time}}(\mean, x') -  V_{\weight_{\time}}(\mean, \mean).
\end{align}
Also, as $\hreg$ takes values in $[0,1]$, $|V_{\weight }(x, x') -  V_{\weight }(x, x)  - (V_{\weight'}(x, x') -  V_{\weight '}(x, x)) | \leq  |\weight - \weight' |$.
Thus:
\begin{align}
\label{gap-inequality}
\abs{\gap_{\weight_{\time}}(\mean) -  \gap(\mean)}
	\leq \weight_{\time}.
\end{align}
As in \eqref{derivationP} one has  for all $k$ and $s$:
\begin{align}
\ddt \pot(\mean_{s+\time_k})&=\langle \payv (\mean_{s+\time_k}),  \BR_{\weight_{s+\time_{k}} }(\mean_{s+\time_k}) - \mean_{s+\time_k}\rangle \end{align}
Thus, using \eqref{gap-equality} then \eqref{gap-inequality} we deduce that:
\begin{align}
\ddt \pot(\mean_{s+\time_k})&\geq \gap_{\weight_{s+\time_{k}}}(\mean_{\time_k+s})- \weight_{s+\time_{k}} \geq \gap(\mean_{\time_k+s})- 2\weight_{s+\time_{k}}.
\end{align}

As the function $\mean \rightarrow \gap (\mean)$ is uniformly continuous and $\time \rightarrow \mean_{\time}$ is Lipschitz, recalling that $\gap (\mean_k) \rightarrow \alpha>0$, we deduce that there is $\delta >0$ and $K_1$ such that for $k\geq K_1$ and all $0\leq s \leq \delta $ one has: 
\begin{align}
\label{gap-inequality2} \gap(\mean_{\time_k+s}) \geq \gap(\mean_{\time_k})- \alpha/4 \geq 3\alpha/4
\end{align}
As $\weight_{\time}$ decreases to $0$ there is $K_2$ such that for all $k\geq K_2$, $\weight_{s+\time_{k}} \leq \alpha/4$.
From this, \eqref{gap-inequality} and \eqref{gap-inequality2}, we obtain that $ \ddt \pot(\mean_{s+\time_k}) \geq \alpha/2$ for all $0\leq s\leq \delta$,  for all $k\geq \max(K_1,K_2)$.
Thus  $\pot(\mean_{\delta+\time_k}) - \pot (\mean_{\time_k})= \int_0^\delta \ddt \pot(\mean_{s+\time_k}) ds \geq \delta \alpha/2$, which contradicts the convergence of  $\pot(\mean_{\time})$.\\

\para{Monotone case}
Using the envelope theorem, and a calculation similar to the above implies that there exists $z(\time) \in \partial \hreg(\mean_{\time})$  such that
\begin{subequations}
\begin{align}
\label{54}
\ddt \gap_{\weight_{\time}}(\mean_{\time}) 
	&= -\langle \payv(\mean_{\time})-\weight_{\time}z(\time) ,\dot \mean_{\time} \rangle + \dot \mean_{\time} \Jac\payv (\mean_{\time}) \dot \mean_{\time}  + \dot \weight_{\time} (\hreg(\mean_{\time})-\hreg(BR_{\weight_{\time} } (\mean_{\time}))
	\\
	&\leq -\langle \payv_{\weight_{\time}}(\mean_{\time}), \dot \mean_{\time} \rangle + \dot \mean_{\time} \Jac\payv (\mean_{\time}) \dot \mean_{\time}  + \dot \weight_{\time} (\hreg(\mean_{\time})-\hreg(BR_{\weight_{\time} } (\mean_{\time}))
	\\
	&\leq -\langle \payv_{\weight_{\time}}(\mean_{\time}), \dot \mean_{\time} \rangle + \dot \weight_{\time} (\hreg(\mean_{\time})-\hreg(BR_{\weight_{\time} } (\mean_{\time}))
	\\
\label{51}
	&= -\langle \payv_{\weight_{\time}}(\mean_{\time}), BR_{\weight_{\time}}- \mean_{\time} \rangle + \dot \weight_{\time} (\hreg(\mean_{\time})-\hreg(BR_{\weight_{\time}} (\mean_{\time})).
\end{align}
\end{subequations}
The first inequality is a consequence of \cref{lem:choice}, which implies that $\langle z(\time), \dot\mean_{\time}\rangle \leq \langle \subsel\hreg(\mean_{\time}), \dot\mean_{\time}\rangle$.
The second inequality is a consequence of $ \dot \mean_{\time} \Jac\payv(\mean_{\time}) \dot \mean_{\time} \leq 0$ as $\game$ is monotone.

We now proceed to prove that $\ddt ( \gap_{\weight_{\time}}(\mean_{\time})+\weight_{\time} ) <0$.
\begin{enumerate}
[\itshape {Case} 1:]
\item
If $\gap_{\weight_{\time}}(\mean_{\time})= 0$ then  $\mean_{\time}= BR_{\weight_{\time}} (\mean_{\time}) = \rNash[\weight_{\time}](\game)$ thus   $\dot \mean_{\time}=0$.
Then \eqref{54} implies  that $\ddt \gap_{\weight_{\time}}(\mean_{\time}) =0$ and thus $\ddt (\gap_{\weight_{\time}}(\mean_{\time}) + \weight_{\time})<0$.

\item
If $\gap_{\weight_{\time}}(\mean_{\time}) >0$ then since $\hreg(\cdot)  \in [0,1]$ and $ \dot \weight_{\time} <0$ we have:
\begin{align}
\label{bounding epsilon}
\dot \weight_{\time} (\hreg(\mean_{\time})-\hreg(BR_{\weight_{\time} } (\mean_{\time}))
	\leq - \dot \weight_{\time} 
\end{align}
and from  \cref{lem:RBRD1} and \eqref{51} we deduce that $\ddt ( \gap_{\weight_{\time}}(\mean_{\time})+\weight_{\time})<0$.
\end{enumerate}

Consequently, $\gap_{\weight_{\time}}(\mean_{\time}) + \weight_{\time}$  converges to $ \delta \geq 0$ and since $\weight_{\time}$ decreases to zero we conclude that $\gap_{\weight_{\time}}(\mean_{\time}) \rightarrow \delta$.
If $\delta>0$, define:
\begin{align}
\Delta
	&\defeq -\langle \payv_{\weight_{\time}}(\mean_{\time}), BR_{\weight_{\time}} (\mean_{\time} )- \mean_{\time} \rangle
	\notag\\
	&\,= -\langle \payv(\mean_{\time}), BR_{\weight_{\time}}(\mean_{\time} )- \mean_{\time} \rangle + \weight_{\time} \langle \subsel\hreg(\mean_{\time}), BR_{\weight_{\time}}(\mean_{\time} )- \mean_{\time} \rangle .
 \end{align}
From \eqref{gap-equality} we further get:
\begin{align}
\Delta &=-\gap_{\weight_{\time}}(\mean_{\time}) + \weight_{\time} (\hreg(\mean_{\time})-\hreg(BR_{\weight_{\time}} (\mean_{\time})) + \weight_{\time} \langle \subsel\hreg(\mean_{\time}), BR_{\weight_{\time}}(\mean_{\time} )- \mean_{\time} \rangle 
\end{align}
As $\gap_{\weight_{\time}}(\mean_{\time}) \rightarrow \delta$, and by the assumptions, $h(\cdot)$ and $\subsel h$ are bounded and $\weight_{\time}$ decreases to $0$, we deduce that for $t$ large enough:
\begin{equation}
\Delta
	= -\langle \payv_{\weight_{\time}}(\mean_{\time}), BR_{\weight_{\time}}(\mean_{\time} )- \mean_{\time} \rangle \leq -\delta/2.
\end{equation}
Consequently, for $t$ large enough we deduce from \eqref{51}, \eqref{bounding epsilon} and the last inequality that
\begin{equation}
\ddt ( \gap_{\weight_{\time}}(\mean_{\time})+\weight_{\time} )\leq -\langle \payv_{\weight_{\time}}(\mean_{\time}), BR_{\weight_{\time}}(\mean_{\time} )- \mean_{\time} \rangle <-\delta/2,
\end{equation}
 which is impossible because $\gap_{\weight_{\time}}(\mean_{\time})+\weight_{\time} \geq 0$.
 Hence $\delta=0$ and so $\mean_{\time} $ converges to $ \Eq(\game)$.
\end{proof}

\subsection{\Acl{DAD} with a constant learning rate}

Before stating our main result we establish some useful lemmas.
Define the regret associated to a trajectory $\strat_{\time}$  and a reference point  $\strat\in \strats$ as:
\begin{equation}
\label{eq:Regret}
\reg_{\strat}(\time)
	= \int_{0}^{\time} \braket{\payv(\strat_{\timealt})}{\strat-\strat_{\timealt}} \dd\timealt
\end{equation}

\begin{lemma}
\label{lem:reg}
Let $\strat_{\time}$ be a solution trajectory of \eqref{eq:DA-cont}.
Then:
\begin{equation}
\reg_{\strat}(\time)
	\leq \max\hreg - \min\hreg.
\end{equation}
\end{lemma}

This is a particular case of a more general bound due to \citet{KM17} (see \cref{lem:reg2} below).
The next more compact proof of this simpler bound follows \citet{BM17} and \citet{Sor21}.

\begin{proof}
Define for $\strat\in\strats$ and $\score\in\R^{\pures}$ the Fenchel coupling
\begin{equation}
\label{eq:Fench}
\fench{\hreg}(x,y)
	= \hreg(x) + \hconj(y) - \braket{x}{y}.
\end{equation}
By the Fenchel-Young inequality, $\fench{\hreg}(x,y) \geq 0$.
Recall that from \eqref{eq:dhconj} and \eqref{eq:DA-int} $\strat_{\time}=\nabla\hconj(\score_{\time})=\choice{\hreg}(\score_t)$ with $\score_{\time}=\int_0^\time \payv(\strat_s)ds $.
Then, following \citet[Lemma~A.2]{MZ19}, a simple derivation gives that for all $\strat\in \strats$:
\begin{equation}
\label{FenchelDerivation}
\ddt\fench{\hreg}(\strat,\score_{\time})
	= \braket{\payv(\strat_{\time})}{\nabla\hconj(\score_{\time}) - \strat}
	= \braket{\payv(\strat_{\time})}{\strat_{\time} - \strat}
	= -\ddt \reg_{\strat}(\time),
\end{equation}
Consequently, following the reasoning of \citet{Sor21}, we get
\begin{equation*}
\reg_{\strat}(\time)
	\leq \fench{\hreg}(\strat,0)
	= \hreg(\strat) + \hconj(0)
	\leq \max\hreg +\hconj(0)
	=  \max\hreg -\min\hreg
	\qedhere
\end{equation*}
\end{proof}

\begin{lemma}
\label{lem:DA-cont}
If a solution trajectory $\strat_{\time} = \eq$  of \eqref{eq:DA-cont} is stationary, then $\eq \in \Eq(\game)$.
\end{lemma}

\begin{proof}
From the previous lemma, for all $t>0$ and $\strat\in\strats$ one has
\begin{equation}
\braket{\payv(\eq)}{\strat-\eq}
	= \frac{\reg_{\strat}(\time)}{\time}
	\leq \frac{\max\hreg - \min\hreg}{\time},
\end{equation}
Letting $\time\to\infty$ implies $\eq\in\Eq(\game)$, as was to be shown.
\end{proof}

\begin{lemma}\label{lemmaR4}
If $\hreg$ satisfies \eqref{eq:strong}, then for all $\strat\in\strats$ and all sequences $\curr[\score]$ in $\R^{\pures}$.
\begin{equation}
\label{R4}
\fench{\hreg}(\strat,\curr[\score])
	\to 0
\implies	
\choice{\hreg}(\curr[\score])
	\to \strat
\end{equation}

\end{lemma}

\begin{proof}
See \cref{lem:Fench2norm} in the appendix.
\end{proof}

We also need to assume the so called ``reciprocity' condition  which requires that for all $\strat\in\strats$ and all sequences $\curr[\score]$ in $\R^{\pures}$:
\begin{equation}
\label{eq:reciprocity}
\tag{Rec}
\choice{\hreg}(\curr[\score])
	\to \strat
	\implies
\fench{\hreg}(\strat,\curr[\score])
	\to 0
\end{equation}
Taken together, \eqref{eq:strong} and \eqref{eq:reciprocity} imply that the sequence $\curr[\strat] = \choice{\hreg}(\curr[\score])$ converges to $\strat\in\strats$ if and only if $\fench{\hreg}(\strat,\curr[\score]) \to 0$.
The terminology ``reciprocity'' is justified by the fact that, under \eqref{eq:reciprocity}, the topology induced by the level sets of the Fenchel coupling becomes equivalent to the ordinary topology on $\strats$;
for a precise statement, see \cref{lem:Fench2norm} in \cref{app:mirror}.
This is a subtle \textendash\ but important \textendash\ requirement, which is fortunately satisfied by all the standard regularization choices on $\strats$, \cf \citep{MerSta18,MZ19}.

We can now state our convergence results for \eqref{eq:DA-cont} when the learning rate is constant.

\begin{theorem}
\label{thm:DA-cont}
Let $\strat_{\time} $ be a solution trajectory of \eqref{eq:DA-cont} with constant learning rate $\learn_{\time} \equiv 1$.
Assume further that one of the following holds:
\begin{enumerate}
[\upshape(\itshape a\hspace*{.5pt}\upshape)]
\item
$\game$ is potential and $\hconj$ is twice differentiable.
\item
$\game$ is strictly monotone, and $\hreg$ satisfies \eqref{eq:strong} and \eqref{eq:reciprocity}.
\end{enumerate}
Then $\strat_{\time}$ converges to $\Eq(\game)$.
\end{theorem}

\begin{proof}
Again, we treat the potential and monotone cases separately.

\para{Potential case}
As $\score_{\time}=\int_0^\time \payv(\strat_{\timealt}) \dd\timealt $ and $\strat_{\time} = \choice{\hreg}(\score_{\time}) = \nabla\hconj(\score_{\time})$ we obtain that: 
\begin{equation}
\dot\strat_{\time}
	= \braket{\payv (\strat_{\time})}{\nabla^{2}\hconj (\score_{\time})}.
\end{equation}
Since $\hconj$ is convex, $\nabla^2\hconj $ is semidefinite positive and so:
\begin{equation}
\ddt\pot(\strat_{\time})
	= \braket{\payv (\strat_{\time})}{\dot\strat_{\time}}
	= \payv(\strat_{\time})\nabla^2\hconj (\score_{\time})\payv (\strat_{\time})
	\geq 0,
\end{equation}
the inequality being strict whenever $\dot\strat_{\time} \neq 0$.
As such, $\strat_{\time}$ converges to $\Eq(\game)$. 

\para{Monotone case}
As $\game$ is strictly monotone, $\Eq(\game)$ is a singleton which coincides with $\eq$ the (necessarily unique) \ac{GESS} of the game.
The calculation in \eqref{FenchelDerivation} gives:
\begin{equation}
\ddt \fench{\hreg}(\eq,\score_{\time})
	= \braket{\payv(\strat_{\time})}{\strat_{\time}-\eq}
	\leq 0
\end{equation}
 As $\eq$ is a \ac{GESS} $\ddt \fench{\hreg}(\eq,\score_{\time})\leq 0$, with equality if and only if $\strat_{\time}=\choice{\hreg}(\score_{\time}) = \eq$.
 Suppose there is $\delta >0$ such that $\braket{\payv(\strat_{\time})}{\strat_{\time}-\eq} < -\delta<0$ for all $t$ large enough.
 Then, $\fench{\hreg}(\eq,\score_{\time})$ goes to $-\infty$ which is impossible because $\fench{\hreg} (\cdot) \geq 0$.
 Thus there is $\time_k \rightarrow \infty$ such that $\braket{\payv(\strat_{\time_k})}{\strat_{\time_k}-\eq} \rightarrow 0$ and so $x_{t_k} = \choice{\hreg}(\score_{t_k}) \rightarrow \eq$.
 By the reciprocity condition \eqref{eq:reciprocity}, we deduce that $\fench{\hreg}(\eq,\score_{t_k}) \rightarrow 0 $ since $\eq$ is a \ac{GESS}.
 As $\ddt \fench{\hreg}(\eq,\score_{\time})\leq 0$ and $\fench{\hreg} (\cdot) \geq 0$, necessarily $\fench{\hreg}(\eq,\score_{\time})$ converges and by the last argument, its limit is zero.
 Using the strong convexity of $h$, \cref{lemmaR4} implies that $\strat_{\time}=\choice{\hreg}(\score_{\time})$ converges to $\eq$. 
\end{proof}

There are several things to note here.
First, the proof of convergence to the unique equilibrium in strictly monotone games (part b of \cref{thm:DA-cont}) extends to all games with a \ac{GESS}.
Second, the same proof implies that in all games, ESS (if exist) are asymptotically stable under \eqref{eq:DA-cont} with constant learning rate (\eg the dynamics converges to an \ac{ESS} if it starts sufficiently close to it).
Third, \cref{thm:DA-cont} concerns the ``day-to-day'' evolution of  \eqref{eq:DA-cont}, not the corresponding empirical mean.
Last, the theorem does not apply to monotone games which are not strict.

To complete the picture, the next proposition shows the convergence of the empirical mean to the set of Nash equilibria in all monotone games.

\begin{proposition}
\label{prop:DA-cont-avg}
Let $\strat_{\time} = \choice{\hreg}(\score_{\time})$ be a solution trajectory of \eqref{eq:DA-cont} with constant learning rate $\learn_{\time} \equiv 1$.
If the underlying game is monotone, then $\mean_{\time}= \frac{1}{t} \int_0^t \strat_s ds$ converges to $\Eq(\game)$.
\end{proposition}

\begin{proof}
By \cref{lem:reg} and monotonicity of the game, for all $\strat\in\strats$:
\begin{equation}
\max\hreg - \min\hreg
	\geq \reg_{\strat}(\time)
	= \int_{0}^{\time} \braket{\payv(\strat_{\timealt})}{\strat - \strat_{\timealt}} \dd\timealt
	\geq \int_{0}^{\time} \braket{\payv(\strat)}{\strat - \strat_{\timealt}} \dd\timealt
\end{equation}
and hence
\begin{equation}
\frac{\max\hreg - \min\hreg}{\time}
	\geq \braket{\payv(\strat)}{\strat - \mean_{\time}}.
\end{equation}
Consequently, $\mean_{\time}$ converges to  $\MVI(\game)$ which coincides with  $\Eq(\game)$ as the game is monotone.
\end{proof}

\subsection{\Acl{DAD} with a variable learning rate}

This subsection extends the results of the previous section on monotone games in two directions:
we consider \acp{GESS} instead of (strictly) monotone games, and we treat dynamics with a variable learning rate.
Other extensions are also possible, but we do not treat them here.

We begin with a technical lemma on the regret incurred by \eqref{eq:DA-cont}:

\begin{lemma}
\label{lem:reg2}
If $\strat_{\time}$ is a solution trajectory of \eqref{eq:DA-cont} then the regret is bounded as follows:
\begin{equation}
\reg_{\strat}(\time)
	\leq \frac{\max\hreg - \min\hreg}{\learn_{\time}}.
\end{equation}
Consequently, if $\time\learn_{\time} \to \infty$ and $\strat_{\time} = \eq$ for all $\time\geq0$, then $\eq \in \Eq(\game)$.
\end{lemma}

\begin{proof}
The bound in the regret follows from \cite[Theorem 4.1]{KM17}, and the implication follows as in \cref{lem:DA-cont} above.
\end{proof}

\cref{lem:reg2} allows to extend  proposition \ref{prop:DA-cont-avg} to variable learning rates.

\begin{proposition}
\label{prop:DA-cont-avg2}
Let $\strat_{\time} = \choice{\hreg}(\score_{\time})$ be a solution trajectory of \eqref{eq:DA-cont}.
If the underlying game is monotone and $\time\learn_{\time} \to \infty$ then $\mean_{\time} = \frac{1}{\time} \int_{0}^{\time} \strat_{\timealt} \dd\timealt$ converges to $\Eq(\game)$.
\end{proposition}

The next result extends part b) of theorem \ref{thm:DA-cont} to variable learning rates.
The proof is a little more complicated.

\begin{theorem}
\label{thm:DA-contb}
Let $\strat_{\time} $ be a solution trajectory of \eqref{eq:DA-cont} with a non-increasing learning rate.
Suppose further that the underlying game $\game$ admits a \ac{GESS} $\eq$.
If the regularizer satisfies \eqref{eq:strong} and \eqref{eq:reciprocity}, and the dynamics' learning rate satisfies
$\lim_{\time\to\infty} \time\learn_{\time} = \infty$,
then $\strat_{\time}$ converges to $\eq$.
\end{theorem}

\begin{proof}
\cite[Eq.~29]{KM17} prove that
\begin{equation}
\ddt \frac{\fench{\hreg}(\eq,\learn_{\time}\score_{\time})}{\learn_{\time}} \leq \braket{\payv(\strat_{\time})}{\strat_{\time}-\eq}-\frac{\dot \eta}{\learn_{\time}^2}(\max h - \min h ).
\end{equation}
Let $\dzone_{\eps} = \setdef{\score\in\R^{\pures}}{\fench{\hreg}(\eq,\score) \leq \eps}$ denote the ``Fenchel zone'' of magnitude $\eps$ with respect to $\eq$.
Then it suffices to show that
\begin{equation}
\label{eq:inzone}
\tag{P}
\textrm{For all $\eps>0$, there exists $T(\eps) > 0$ such that $\learn_{\time}\score_{\time} \in \dzone_{\eps}$ for all $t\geq T(\eps)$.}
\end{equation}
Actually this means that $\fench{\hreg}(\eq, \learn_{\time}\score_{\time}) \to 0$.
By \eqref{eq:strong}, \eqref{R4} holds and thus  $\strat_{\time}$ converges to $\eq$.
The proof of \eqref{eq:inzone} follows from the following two claims:

\para{Claim 1}  $\learn_{\time}\score_{\time} \in \dzone_{\eps}$ infinitely often.
\begin{proof}
\renewcommand{\qedsymbol}{\S}
If  $\learn_{\time}\score_{\time} \notin \dzone_{\eps}$, then  by \eqref{eq:reciprocity},  there exists  a neighbourhood $\mathcal{U} (\eps)$ of $\eq$ such that $\strat_{\time} \notin \mathcal{U}(\eps)$.
By definition of \ac{GESS}, there is $c(\eps)>0$ such that $\braket{\payv(\strat_{\time})}{ \strat_{\time}-\eq} < -2c(\eps)$.
As  $\frac{\dot \learn_{\time}}{\learn_{\time}^2} \rightarrow 0$, there is  $T_1(\eps)$ such that for all $t\geq T_1(\eps) $ we have $\frac{\dot \learn_{\time}}{\learn_{\time}^2} (\max h - \min h )  < c(\eps)$.
Consequently by contradiction to the claim,  $\ddt \frac{\fench{\hreg}(\eq,\learn_{\time}\score_{\time})}{\learn_{\time}}<-c(\eps)$ for all $\time$ large enough and so $\frac{\fench{\hreg}(\eq,\learn_{\time}\score_{\time})}{\learn_{\time}}$ goes to $-\infty$.
Since $\fench{\hreg}(\eq,\learn_{\time}\score_{\time}) \geq 0$ this is impossible.
\end{proof}

\para{Claim 2}
There exists $T_2(\eps) \geq T_1(\eps/2)$ such that if $t\geq  T_2(\eps) $ and  $\learn_{\time}\score_{\time} \in \dzone_{\eps}$ then there exists $\tau>0$ such that $\learn_{\timealt}\score_{\timealt} \in \dzone_{\eps}$  for all $s\in [t, t+\tau]$.

\begin{proof}
\renewcommand{\qedsymbol}{\S}
We need to consider two cases separately:
\begin{enumerate}
[left=0pt,label={\emph{Case} \itshape\alph*\upshape).}]
\item
If $\learn_{\time} \score_{\time} \in \dzone_{\eps/2} $ then obviously, there is $\tau>0$ such that  $\learn_{\timealt} \score_{\timealt} \in \dzone_{\eps} $ for all $s\in [t, t+\tau]$.
\item
If $\learn_{\time} \score_{\time}  \in \dzone_{\eps} $ but $\learn_{\time} \score_{\time}  \notin \dzone_{\eps/2} $ then, from claim 1, $\ddt \frac{\fench{\hreg}(\eq,\learn_{\time}\score_{\time})}{\learn_{\time}}<-c(\eps/2) < 0$.
Thus there exists $\tau>0$ such that for all $s\in [t, t+\tau]$: $\frac{\fench{\hreg}(\eq,\learn_{\timealt}\score_{\timealt})}{\learn_{\timealt}} \leq \frac{\fench{\hreg}(\eq,\learn_{\time}\score_t)}{\learn_{\time}}$.
As $\learn_{\time}\score_{\time} \in \dzone_{\eps} $, $\fench{\hreg}(\eq,\eta_{\time}\score_{\time}) \leq \eps$ and since  $\learn_{\timealt} \leq \learn_{\time}$, we deduce that $\fench{\hreg}(\eq,\learn_{\timealt}\score_{\timealt}) \leq \eps $.
Thus $\learn_{\timealt} \score_{\timealt}  \in \dzone_{\eps} $  for all $s\in [t, t+\tau]$.
\qedhere
\end{enumerate}
\end{proof}

From Claims 1 and 2, there exists $t \geq \max(T_2(\eps), T(\eps))$,  such that $\learn_{\time} \score_t$ belongs to $\dzone_{\eps}$ and remains in it forever.
\end{proof}

\section{Analysis and results in discrete time}
\label{sec:disc}

In this section, we continue with the analysis of the discrete-time dynamics presented in \cref{sec:dyn-disc},
starting with \acl{FP} and its variants in \cref{sec:FP-disc} before moving on to the dynamics of \acl{DA} in \cref{sec:DA-disc}.

\subsection{\Acl{FP} and \acl{RFP}}
\label{sec:FP-disc}

The convergence of \eqref{eq:FP} and \eqref{eq:RFP} can be obtained from the continuous-time analysis of \cref{sec:cont} using the theory of stochastic approximation \citep{Ben99,BHS05,BHS06}.
Informally, this theory links the asymptotic behavior of differential inclusions of the form
\begin{equation}
\label{eq:ODI}
\dot\strat_{\time}
	\in \oper(\strat_{\time})
\end{equation}
to the limit sets of discrete-time processes satisfying the basic recursion
\begin{equation}
\next[\state] -  \curr
	\in \next[\step] \oper(\curr),
\end{equation}
where $\oper\from\strats\too\R^{\pures}$ is an upper semicontinuous,
compact-convex-valued correspondence on $\strats$, and $\step_{\run} > 0$ is a vanishing step-size sequence.

In our specific context, the role of $\oper$ is to be played by the best-reply correspondence $\BR$ for \eqref{eq:FP}, and its regularized variant $\rBR{\weight\hreg}$ for \eqref{eq:RFP}.
Concretely, we have:

\begin{proposition}
\label{prop:APT}
Let $\curr[\avg] = (1/\run) \sum_{\runalt=\start}^{\run} \iter$, $\run=\running$, be the empirical frequency of play under a sequence of population states $\curr\in\strats$.
Then:
\begin{enumerate}
[\upshape(\itshape i\hspace{1pt}\upshape)]
\item
If $\curr$ is generated by \eqref{eq:FP}, $\curr[\avg]$ is an \acl{APT} of \eqref{eq:BRD}.
\item
If $\curr$ is generated by \eqref{eq:RFP}, $\curr[\avg]$ is an \acl{APT} of \eqref{eq:RBRD}.
\end{enumerate}
\end{proposition}

Informally, the notion of an \acf{APT} means that $\curr[\avg]$ asymptotically tracks the pseudo-flow of \eqref{eq:BRD}/\eqref{eq:RBRD} with arbitrary accuracy over windows of arbitrary length;
for a precise definition, see \citet[Definitions I and 4.1]{BHS05}.%
The proof of \cref{prop:APT} follows immediately from the general theory of \citet[Propositions 1.3, 1.4, and Theorem 4.2]{BHS05}, so we do not present it here.

In view of this ``asymptotic approximation'' property, it is plausible to expect that $\curr[\avg]$ exhibits the same asymptotic behavior as the corresponding continuous-time dynamics.
Our next result makes this intuition precise in two classes of games:%
\begin{enumerate}
\item
Monotone games, \ie games that satisfy \eqref{eq:monotone}.
\item
Potential games for which the $\weight$-regularized potential $\pot_{\weight} = \pot - \weight\hreg$, $\weight\geq0$, satisifies the Sard-type condition:
\begin{equation}
\label{eq:Sard}
\tag{S$_{\eps}$}
\text{$\pot_{\weight}(\rEq{\weight}(\game))$ has empty interior}.
\end{equation}
\end{enumerate}
Our main results for \eqref{eq:FP} and \eqref{eq:RFP} may then be stated as follows:

\begin{theorem}
\label{thm:FP-disc}
Suppose that $\game$ satisfies \eqref{eq:monotone} or \eqref{eq:Sard} for $\weight=0$.
Then the empirical frequency of play $\curr[\avg]$ under \eqref{eq:FP} converges to $\Eq(\game)$.
\end{theorem}

\begin{theorem}
\label{thm:RFP-disc}
Suppose \eqref{eq:selection} holds and $\game$ satisfies \eqref{eq:monotone} or \eqref{eq:Sard} for some $\weight>0$.
Then the empirical frequency of play $\curr[\avg]$ under \eqref{eq:RFP} converges to $\rEq{\weight}(\game)$.
\end{theorem}

\begin{proof}[Sketch of proof]
The proof of \cref{thm:FP-disc,thm:RFP-disc} follows a similar technical path starting with the fact that $\curr[\avg]$ is an \ac{APT} of \eqref{eq:BRD}\,/\,\eqref{eq:RBRD} respectively (\cf \cref{prop:APT} above).
In both cases, Theorem 4.3 of \citet{BHS05} shows that the accumulation points of an \ac{APT} lie in an internally chain transitive set of the underlying continuous-time dynamics.%
\footnote{An internally chain transitive set is a compact invariant set that contains no proper attractors \citep{Con78}.}
The claims of \cref{thm:FP-disc} may then be proved as follows:
\begin{enumerate}
\item
For games satisfying \eqref{eq:monotone}, invoke \cref{thm:BRD} from \cref{sec:cont} and Theorem 3.23 and Proposition 3.25 of \citet{BHS05}.
\item
For games satisfying \eqref{eq:Sard} with $\weight=0$, invoke \cref{thm:BRD} from \cref{sec:cont} and Proposition 3.27 of \citet{BHS05}.
\end{enumerate}
Finally, for $\weight>0$, the proof of \cref{thm:RFP-disc} is entirely analogous, and simply invokes \cref{thm:RBRD} in lieu of \cref{thm:BRD} in the above chain of implications.
\end{proof}

\subsection{\Acl{DA}}
\label{sec:DA-disc}

We now proceed to state and prove our main convergence result for the recursion \eqref{eq:DA}.
To do so, we will require the ``reciprocity'' condition \eqref{eq:reciprocity} that was introduced for the study of the continuous-time dynamics \eqref{eq:DA-cont}, namely
\begin{equation*}
\tag{\ref*{eq:reciprocity}}
\choice{\hreg}(\curr[\score])
	\to \strat
	\implies
\fench{\hreg}(\strat,\curr[\score])
	\to 0
\end{equation*}
for all $\strat\in\strats$ and all sequences $\curr[\score]$ in $\R^{\pures}$.
We then have:

\begin{theorem}
\label{thm:DA-disc}
Suppose that \eqref{eq:DA} is run with a regularizer satisfying \eqref{eq:strong} and \eqref{eq:reciprocity}, and a vanishing non-increasing learning rate $\curr[\learn] \searrow 0$ such that
\begin{equation}
\label{eq:learn}
\lim_{\run\to\infty}
	\parens*{ \frac{1}{\curr[\learn]} - \frac{1}{\prev[\learn]} }
	= 0.
\end{equation}
Then:
\begin{itemize}
\addtolength{\itemsep}{\smallskipamount}

\item
If $\eq$ is a \ac{GESS} of $\game$,
the sequence of population states $\curr$ converges to $\eq$.

\item
If $\eq$ is an \ac{ESS} of $\game$,
the same conclusion holds provided that
$\max\{\curr[\learn],\curr[\learn]^{-1} - \prev[\learn]^{-1} \}$ is small enough
and
\eqref{eq:DA} is initialized sufficiently close to $\eq$.

\end{itemize}
\end{theorem}

\begin{corollary}
Suppose that $\game$ is strictly monotone
and
\eqref{eq:DA} is run with
a regularizer satisfying \eqref{eq:strong} and \eqref{eq:reciprocity},
and
a learning rate of the form $\curr[\learn] \propto 1/\run^{\pexp}$ for $\pexp\in(0,1)$.
Then the sequence of population states induced by \eqref{eq:DA} converges to the \textpar{necessarily unique} equilibrium of $\game$.
\end{corollary}

Our proof relies on the use of a suitable ``energy inequality'', provided here by a deflated variant of the Fenchel coupling:

\begin{lemma}
\label{lem:template}
Fix a population state $\strat\in\strats$,
and let
\begin{equation}
\label{eq:Fench-deflated}
\curr[\lyap]
	\defeq \frac{1}{\curr[\learn]} \fench{\hreg}(\strat,\curr[\learn]\curr[\aggpay])
	= \frac{1}{\curr[\learn]} \bracks{
		\hreg(\strat)
		- \hreg(\next)
		+ \curr[\learn]\braket{\curr[\aggpay]}{\next - \strat}
		},
\end{equation}
with $\curr[\aggpay]$ as in \eqref{eq:DA}.
Then, for all $\run=\running$, we have:
\begin{equation}
\label{eq:template}
\curr[\lyap]
	\leq \prev[\lyap]
		+ \braket{\payv(\curr)}{\curr - \strat}
		+ \bracks{\hreg(\strat) - \min\hreg} \, \curr[\diff]
		+ \frac{\fench{\hreg}(\curr,\prev[\learn]\curr[\aggpay])}{\prev[\learn]},
\end{equation}
where $\curr[\diff] = 1/\curr[\learn] - 1/\prev[\learn]$.
If, in addition, $\hreg$ satisfies \eqref{eq:strong}, we have
\begin{equation}
\label{eq:template-strong}
\curr[\lyap]
	\leq \prev[\lyap]
		+ \braket{\payv(\curr)}{\curr - \strat}
		+ \bracks{\hreg(\strat) - \min\hreg} \, \curr[\diff] 
		+ \frac{\prev[\learn]}{2\hstr} \dnorm{\payv(\curr)}^{2}.
\end{equation}
\end{lemma}

Energy inequalities of the form \eqref{eq:template} are standard in the analysis of \acl{DA} schemes, and they go back at least as far as \citet{Nes07,Nes09} (who introduced the method).
The specific formulation above follows \citet[Lemma.~C.1]{HMMR20}, so we omit the proof;
for a range of similar computations, see \citet[Lemma~2.20]{SS11}, \citet[p.~42]{Kwo16} and \citet{Sor21}.

Now, to move forward with the proof of \cref{thm:DA-disc}, we will require two complementary arguments.
The first is a stability result:
we show below that, if $\curr$ is sufficiently close to an \acl{ESS} $\eq$ and the method's learning rate is ``small enough'', then $\next$ remains close to $\eq$.
To state this precisely, let $\dzone_{\eps} = \setdef{\score\in\R^{\pures}}{\fench{\hreg}(\eq,\score) \leq \eps}$ denote the ``Fenchel zone'' of magnitude $\eps$ with respect to $\eq$,
and let $\zone_{\eps} = \choice{\hreg}(\dzone_{\eps}) = \setdef{\strat = \choice{\hreg}(\score)}{\score\in\dzone_{\eps}}$ denote the image of $\dzone_{\eps}$ under $\choice{\hreg}$.
We then have:

\begin{lemma}
\label{lem:stability}
Let $\eq \in \ESS(\game)$ and assume that \eqref{eq:strong} and \eqref{eq:reciprocity} hold.
Then, for all sufficiently small $\eps>0$, we have the following implication \textpar{valid for all $\run=\running$}:
if $\curr\in\zone_{\eps}$ and $\max\{\curr[\learn],\curr[\diff]\}$ is small enough, we also have $\next\in\zone_{\eps}$.
\end{lemma}

The second component of the proof of \cref{thm:DA-disc} is a subsequence extraction argument:
if the iterates of \eqref{eq:DA} lie in a neighborhood of $\eq$ where \eqref{eq:ESS} holds,
there exists a subsequence of $\curr$ that converges to $\eq$.

\begin{lemma}
\label{lem:subsequence}
Let $\eq \in \ESS(\game)$, let $\nhd$ be the neighborhood of validity of \eqref{eq:ESS}, and assume further that \eqref{eq:strong} and \eqref{eq:learn} hold.
If $\curr\in\nhd$ for all sufficiently large $\run$, we have $\liminf_{\run\to\infty} \norm{\curr - \eq} = 0$.
\end{lemma}

We prove these two ancillary results below.

\begin{proof}[Proof of \cref{lem:stability}]
Let $\nhd$ be the neighborhood whose existence is guaranteed by \eqref{eq:ESS}, \ie $\braket{\payv(\strat)}{\strat - \eq} < 0$ for all $\strat\in\nhd$, $\strat\neq\eq$.
By the relation between the Fenchel coupling and the ordinary norm topology on $\strats$ (\cf \cref{lem:Fench2norm} in \cref{app:mirror}),
it follows that $\zone_{\eps} \subseteq \nhd$ if $\eps$ is taken small enough;
we will assume this to be the case throughout the rest of this proof.
Moreover, by \eqref{eq:reciprocity} and the continuity of $\pay$, there exists some $\const \equiv \const(\eps)$ such that $\braket{\payv(\strat)}{\strat - \eq} \leq - \const < 0$ for all $\strat \in \zone_{\eps} \setminus \zone_{\eps/2}$.

With all this in hand, assume that $\curr \in \zone_{\eps}$;
we will show that $\next \in \zone_{\eps}$ provided that $\max\{\curr[\learn],\curr[\diff]\}$ is small enough.
To do so, let $\curr[F] = \fench{\hreg}(\eq,\curr[\learn]\curr[\aggpay])$;
then, by substituting $\strat\gets\eq$ in the template inequality \eqref{eq:template-strong} of \cref{lem:template}, we obtain:

\begin{enumerate}
[left=1em,label={\bfseries Case \arabic*}.]
\addtolength{\itemsep}{\smallskipamount}

\item
If $\curr\in\zone_{\eps} \setminus \zone_{\eps/2}$, then
\begin{equation}
\curr[F]
	\leq \prev[F]
		- \curr[\learn] \bracks*{\const - \curr[\diff] \depth - \prev[\learn] \frac{\vbound^{2}}{2\hstr}},
\end{equation}
where $\depth = \max\hreg - \min\hreg$ and $\vbound = \max_{\strat\in\strats} \dnorm{\payv(\strat)}$.
Hence, if $\curr[\diff] \leq \const/(2\depth)$ and $\prev[\learn] \leq \const\hstr/\vbound^{2}$, we conclude that $\curr[F] \leq \prev[F]$.

\item
If $\curr\in\zone_{\eps/2}$, then
\begin{equation}
\curr[F]
	\leq \frac{\eps}{2}
		+ \curr[\diff] \depth + \prev[\learn] \frac{\vbound^{2}}{2\hstr},
\end{equation}
so $\curr[F] \leq \eps$ provided that $\curr[\diff] < \eps/(4\depth)$ and $\prev[\learn] \leq \eps \hstr / (2\vbound^{2})$.
\end{enumerate}
In both cases, we conclude that $\next = \choice{\hreg}(\curr[\learn]\curr[\aggpay])\in\zone_{\eps}$, as claimed.
\end{proof}

\begin{proof}[Proof of \cref{lem:subsequence}]
Assume to the contrary that $\liminf_{\run\to\infty} \norm{\curr - \eq} > 0$.
Then, by assumption, there exist $\eps>0$ and $\run_{0} \equiv \run_{0}(\eps)$ such that $\curr\in\nhd$ and $\norm{\curr - \eq} \geq \eps$ for all $\run\geq\run_{0}$.
Since $\eq\in\ESS(\game)$, there exists a positive constant $\const>0$ such that $\braket{\payv(\curr)}{\curr - \eq} \leq -\const < 0$ for all $\run\geq\run_{0}$.
With this in mind, substituting $\strat \gets \eq$ in the template inequality \eqref{eq:template-strong} and telescoping yields
\begin{align}
\label{eq:lyap-subseq}
\curr[\lyap]
	&\leq \lyap_{\run_{0}}
		- \const(\run - \run_{0})
		+ \parens*{\frac{1}{\curr[\learn]} - \frac{1}{\learn_{\run_{0}}}} \depth
		+ \frac{\vbound^{2}}{2\hstr} \sum_{\runalt=\run_{0}}^{\run} \iter[\learn]
	\notag\\
	&\leq \lyap_{\run_{0}} + \const \run_{0} - \depth/\run_{0}
		- \run \bracks*{\const - \frac{1}{\run \curr[\learn]} - \frac{1}{\run} \sum_{\runalt=\run_{0}}^{\run} \iter[\learn]}.
\end{align}
To proceed, note that the first part of \eqref{eq:learn} gives
\begin{equation}
\lim_{\run\to\infty} \frac{\parens*{1/\curr[\learn] - 1/\prev[\learn]}}{\run - (\run-1)}
	= \lim_{\run\to\infty} \parens*{\frac{1}{\curr[\learn]} - \frac{1}{\prev[\learn]}}
	= \lim_{\run\to\infty} \curr[\diff]
	= 0,
\end{equation}
so $\lim_{\run\to\infty} 1/(\run\curr[\learn]) = 0$ by the Stolz\textendash Cesàro theorem.
By the second part of \eqref{eq:learn}, we also have $(1/\run) \sum_{\runalt=\run_{0}}^{\run} \iter[\learn] = 0$.
Hence, by letting $\run\to\infty$ in \eqref{eq:lyap-subseq}, we get $\lim_{\run\to\infty} \curr[\lyap] = -\infty$, a contradiction which completes our proof.
\end{proof}

The proof of our main convergence results follows by a tandem application of the above.

\begin{proof}[Proof of \cref{thm:DA-disc}]
We consider the two cases separately.

\para{Case 1: $\eq\in\GESS(\game)$}
Let $\eps>0$ be arbitrary.
By \eqref{eq:learn}, it follows that $\lim_{\run\to\infty} \curr[\learn] = \lim_{\run\to\infty} \curr[\diff] = 0$, and, by \cref{lem:subsequence}, there exists a subsequence $\state_{\run_{\runalt}}$ of $\curr$ converging to $\eq$.
This means that, in the long run, $\curr[\learn]$ and $\curr[\diff]$ become arbitrarily small, and $\curr$ comes arbitrarily close to $\eq$.
Hence, by applying \cref{lem:stability} inductively, we conclude that $\curr \in \zone_{\eps}$ for all sufficiently large $\run$.
Since $\eps>0$ has been chosen arbitrarily, this implies that $\curr$ converges to $\eq$, as claimed.

\para{Case 2: $\eq\in\ESS(\game)$}
Let $\nhd$ be the neighborhood whose existence is guaranteed by the evolutionary stability condition \eqref{eq:ESS-VI} for $\eq$.
Then, by \cref{lem:stability}, if \eqref{eq:DA} is initialized sufficiently close to $\eq$ and $\max\{\curr[\learn],\curr[\diff]\}$ is sufficiently small, we will have $\curr\in\nhd$ for all $\run$.
Hence, by
\cref{lem:subsequence}, we conclude that $\liminf_{\run\to\infty} \norm{\curr - \eq} = 0$.
Our proof then follows by the same inductive application of \cref{lem:stability} as in the case of a \ac{GESS} above.
\end{proof}

\para{Time averages under \acl{DA}}
We close this section with a result on the empirical frequency of play under \eqref{eq:DA}.
Indeed, telescoping \eqref{eq:template} trivially yields the bound:
\begin{equation}
\sum_{\runalt=\start}^{\run} \braket{\payv(\iter)}{\strat - \iter}
	\leq \frac{\max\hreg - \min\hreg}{\curr[\learn]}
	+ \frac{1}{2\hstr} \sum_{\runalt=\start}^{\run} \beforeiter[\learn] \dnorm{\payv(\iter)}^{2}.
\end{equation}
This regret-type bound echoes standard results in the literature \textendash\ see \eg \citet[Theorem 1]{Nes09} \textendash\ and leads to the following time-average convergence result:

\begin{proposition}
\label{prop:DA-disc-avg}
Suppose that \eqref{eq:DA} is run with
a regularizer satisfying \eqref{eq:strong}
and
a learning rate of the form $\curr[\learn] \propto 1/\sqrt{\run}$.
If $\game$ is monotone, the empirical frequency of play $\curr[\avg]$ under \eqref{eq:DA} converges to $\Eq(\game)$.
\end{proposition}

The proof of \cref{prop:DA-disc-avg} mimics that of \cref{prop:DA-cont-avg} so we omit it.
We only note that, in contrast to \cref{thm:DA-disc}, \cref{prop:DA-disc-avg} does not require strict monotonicity or reciprocity;
however, it concerns the cruder, average frequency of play, so it provides a considerably weaker guarantee in this regard.

\section{Extensions and concluding remarks}
\label{sec:extensions}

In this concluding section, we provide a series of extensions and remarks that would have otherwise disrupted the flow of the rest of our paper.

\subsection{Positive correlation}

\citet{San01} introduced the so-called ``\emph{positive correlation}'' condition
\begin{equation}
\label{eq:corr}
\tag{PC}
\dot \strat_{\time}
	\neq 0
	\implies \braket{\payv(\strat_{\time})}{\dot\strat_{\time}} > 0.
\end{equation}
In the case of potential games, \eqref{eq:corr} implies that
\begin{equation}
\ddt \pot(\strat_{\time})
	= \braket{\payv(\strat_{\time})}{\dot\strat_{\time}}
	> 0
\end{equation}
whenever $\dot\strat_{\time} \neq 0$.
This in turn implies the  convergence of the trajectories to the set of rest points of the dynamics \citep[Theorem 7.1.2]{San10}.
Some of our results in \cref{sec:cont} for potential games can also be obtained by establishing \eqref{eq:corr} for the dynamics under study.

\subsection{No-regret dynamics}

Our second remark concerns continuous- or discrete-time processes for which the average regret vanishes, \ie for all $\strat\in\strats$ we have:
\begin{subequations}
\label{eq:no-reg}
\begin{alignat}{2}
\limsup_{\time\to\infty}
	\frac{1}{\time}\int_{0}^{\time} \braket{\payv(\strat_{\timealt})}{\strat - \strat_{\timealt}} \dd\timealt
	&= 0
	&\quad
	&\text{in continuous time}
	\\
\limsup_{\run\to\infty}
	\frac{1}{\run}\sum_{\runalt=\start}^{\run} \braket{\payv(\iter)}{\strat - \iter}
	&= 0
	&\quad
	&\text{in discrete time}
\end{alignat}
\end{subequations}
As we saw in \cref{prop:DA-cont-avg,prop:DA-disc-avg}, the dynamics of \acl{DA} have the no-regret property \eqref{eq:no-reg} in both continuous and discrete time.
Looking more closely at the proofs of \cref{prop:DA-cont-avg,prop:DA-disc-avg}, we conclude that, in monotone games, the empirical frequency of play under any dynamical process that satisfies \eqref{eq:no-reg} converges to the game's set of equilibria.
This property forms the basis of many algorithms and dynamics designed to solve variational inequalities and equilibrium problems \citep{Nem04,Nes09,JNT11};
the above provides a game-theoretic interpretation of this property in the context of population games.

\subsection{Extensions}

We also note that our results can be extended to the following settings:
\begin{enumerate}

\item
If $\strats$ is a convex compact subset of $\R^{\pures}$ and $\payv\from\strats\to\R^{\pures}$ is a vector field on $\strats$, we maintain the same results for convergence to solutions of \eqref{eq:SVI} and \eqref{eq:MVI};
however, the population game interpretation disappears.

\item
This interpretation is recovered in multi-population games:
if there are several player populations indexed by $\play=1,\dotsc,\nPlayers$, and if $\pures_{\play}$ denotes the set of pure strategies of the $\play$-th population and $\pay_{\play\pure}\from\prod_{\play=1}^{\nPlayers} \simplex(\pures_{\play}) \to \R$ is the payoff function of $\pure$-strategists in the $\play$-th population, it suffices to set $\strats = \prod_{\play} \simplex(\pures_{\play})$ and $\payv(\strat) = (\pay_{\play\pure_{\play}}(\strat))_{\pure_{\play}\in\pures_{\play},\play\in\players}$.

\item
Our theory can also be extended to normal form games with a finite number of players and smooth concave payoff functions $g_{\play}$, $\play=1,\dotsc,\nPlayers$.
In this case, the components of the corresponding vector field are $G_{\play\pure}(\strat) = \frac{\pd g_{\play}}{\pd\strat_{\play\pure}}$, and the notion of evolutionary stability boils down to that of \emph{variational stability};
for a detailed treatment, see \citet{SW16} and \citet{MZ19}.
\end{enumerate}

\subsection{Links with \aclp{VI} and further extensions}

Taking a more general point of view that focuses on the vector field $\payv$, we can also consider the following framework that brings to the forefront the associated \aclp{VI}.
Let $\points_{\play}$, $\play = 1,\dotsc,\nPlayers$ be a convex compact subset of a topological vector space $\vecspace_{\play}$ with dual $\dspace_{\play}$,
write $\vecspace = \prod_{\play}\vecspace_{\play}$ and $\points = \prod_{\play}\points_{\play}$,
and
let $\payv_{\play} \from \points\to\dspace_{\play}$ be a collection of continuous maps, $\play=1,\dotsc,\nPlayers$.
Finally, set $\payv(\point) = (\payv_{1}(\point),\dotsc,\payv_{\nPlayers}(\point))$ and, as usual, write
\begin{equation}
\braket{\payv(\point)}{\tvec}
	= \sum_{\play=1}^{\nPlayers} \braket{\payv_{\play}(\point)}{\tvec_{\play}}
\end{equation}
for the dual pairing between $\payv(\point) \in \dspace = \prod_{\play}\dspace_{\play}$ and $\tvec = (\tvec_{1},\dotsc,\tvec_{\nPlayers}) \in \vecspace \defeq \prod_{\play} \vecspace_{\play}$.

We may then define $\SVI(\points,\payv)$ as the set of $\sol\in\points$ such that
\begin{equation}
\braket{\payv(\sol)}{\point-\sol}
	\leq 0,
	\quad
	\text{for all $\point\in\points$}
\end{equation}
and $\MVI(\points,\payv)$ as the set of $\sol\in\points$ such that
\begin{equation}
\braket{\payv(\point)}{\point-\sol}
	\leq 0
	\quad
	\text{for all $\point\in\points$}.
\end{equation}
The counterparts for $\ESS$ and $\GESS$ may also be defined similarly (though the link with population games is no longer present).

\begin{example}
\label{ex:cont}
Consider a strategic game with $\nPlayers$ players, compact metric strategy spaces $\pures_{\play} \subseteq \R^{\vdim}$ and continuous payoff functions $\pay_{\play} \from \pures \to \R$, $\play=1,\dotsc,\nPlayers$, with $\pures \defeq \prod_{\playalt}\pures_{\playalt}$.
Introduce the mixed extension of the game as usual:
for a probability distribution (mixed strategy profile) $\chi = (\chi_{1},\dotsc,\chi_{\nPlayers}) \in \strats$ with $\chi_{\play} \in \strats_{\play} \defeq \simplex(\pures_{\play})$, let $\pay_{\play}$ be extended as
\begin{equation}
\pay_{\play}(\chi)
	= \int_{\pures} \pay_{\play}(\pure_{1},\dotsc,\pure_{\nPlayers})
	\,\dd\chi_{1}(\pure_{1}) \dotsm \dd\chi_{\nPlayers}(\pure_{\nPlayers}).
\end{equation}
Then, introduce $\payv_{\play}(\chi) = \pay_{\play}(\cdot,\chi_{-\play})$ so that
\begin{equation}
\braket{\payv_{\play}(\chi)}{\alt\chi_{\play}}
	= \payv_{\play}(\alt\chi_{\play};\chi_{-\play})
\end{equation}
for all $\alt\chi_{\play}\in\strats_{\play}=\simplex(\pures_{\play})$.
In this way, the equilibrium condition for $\eq[\chi]\in\strats$ becomes
\begin{equation}
\braket{\payv(\eq[\chi])}{\chi - \eq[\chi]}
	= \sum_{\play=1}^{\nPlayers} \braket{\payv_{\play}(\eq[\chi])}{\chi_{\play} - \eq[\chi]_{\play}}
	\leq 0
	\quad
	\text{for all $\chi\in\strats$},
\end{equation}
which is again an instance of $\SVI(\points,\payv)$.
\hfill
\endenv
\end{example}

In this setting, $\payv$ is monotone (dissipative) if
\begin{equation}
\braket{\payv(\point) - \payv(\pointalt)}{\point - \pointalt}
	\leq 0
	\quad
	\text{for all $\point,\pointalt\in\points$}
\end{equation}
As in the case of \cref{sec:prelims}, we have $\SVI(\points,\payv) = \MVI(\points,\payv)$ if $\payv$ is monotone;
furthermore, mutatis mutandis, all the properties presented in \cref{sec:classes} extend to this general setup.

The best-response map associated to $\payv$ is defined again as
 \begin{equation}
\BR(\point)
	= \argmax_{\pointalt\in\points} \braket{\payv(\point)}{\pointalt}.
\end{equation}
In addition to monotonicity, we will also now make the ``\emph{unique best response}'' assumption
\begin{equation}
\label{eq:unique}
\tag{U}
\text{$\BR(\point)$ is a singleton for all $\point\in\points$.}
\end{equation}

\begin{proposition}
\label{prop:singleton}
If $\payv$ is  monotone and \eqref{eq:unique} holds, the set $E \defeq \MVI(\points,\payv) = \SVI(\points,\payv)$ is a singleton.
\end{proposition}

\begin{proof}
If $\point,\pointalt\in E$, we have $\braket{\payv(\point)}{\pointalt - \point} \leq 0$ and $\braket{\payv(\pointalt)}{\point-\pointalt} \leq 0$, so
\begin{equation}
\braket{\payv(\pointalt) - \payv(\point)}{\pointalt - \point}
	\geq 0
\end{equation}
implying first $\braket{\payv(\pointalt) - \payv(\point)}{\pointalt - \point} = 0$ by monotonicity and then $\braket{\payv(\pointalt)}{\pointalt - \point} = 0$.
Thus, $\point = \pointalt$ by \eqref{eq:unique}.
\end{proof}

The regularized best response map is likewise defined as
\begin{equation}
\rBR{\weight\hreg}(\point)
	= \argmax_{\pointalt\in\points} \{\braket{\payv(\point)}{\pointalt} - \eps\hreg(\pointalt) \}
\end{equation}
where $\hreg\from\points\to\R$ satisfies the same axioms as \cref{def:reg}.
Note that $\rBR{\weight\hreg}$ satisfies \eqref{eq:unique} by construction.

Now, using these best-response maps, we may extend the definition of the dynamics \eqref{eq:FP}, \eqref{eq:RFP} and \eqref{eq:DA} to the more general setup considered here in the obvious way.
We then have:

\begin{proposition}
Suppose that \eqref{eq:unique} holds and let $\curr$, $\run=\running$, be a sequence of states generated by \eqref{eq:FP}.
If the accumulation points of $\curr[\avg]$ lie in $\SVI(\points,\payv)$,
the same holds for the accumulation points of $\curr$.
\end{proposition}

\begin{proof}
Let $\eq$ be an accumulation point of $\next = \BR(\curr[\avg])$, and let $\state_{\run_{\runalt}+1}$ be a subsequence converging to $\eq$.
By descending to a further subsequence if necessary, we may assume without loss of generality that $\avg_{\run_{\runalt}}$ converges to some limit point $\bar\point$.
Furthermore, by the definition of \eqref{eq:FP}, we have:
\begin{equation}
\braket{\payv(\bar\state_{\run_{\runalt}})}{\state_{\run_{\runalt}+1} - \point}
	= \braket{\payv(\bar\state_{\run_{\runalt}})}{\BR(\bar\state_{\run_{\runalt}}) - \point}
	\geq 0
	\quad
	\text{for all $\point\in\points$}.
\end{equation}
Hence, taking the limit $\runalt\to\infty$, we get
\begin{equation}
\braket{\payv(\bar\point)}{\sol - \point}
	\geq 0
	\quad
	\text{for all $\point\in\points$},
\end{equation}
so $\eq \in \BR(\bar\point)$.
Since $\bar\point \in \SVI(\points,\payv)$ by assumption, we deduce further that $\bar\point \in \BR(\bar\point)$.
Hence, by \eqref{eq:unique}, we conclude that $\bar\point=\sol$, and our assertion follows.
\end{proof}

Obviously, under \eqref{eq:unique}, a similar property holds for \acl{RFP}.
 
%

\subsection{Preview for Part II}

Part II of our paper \citep{HLMS21b} deals with games with continuous action spaces, in the spirit of \cref{ex:cont} (but with a continuous population of nonatomic players).
The first step is to define the \acl{SVI} when $\pures$ is a compact metric space,
$\Theta$ is an appropriate compact convex subset of $\Delta(\pures)$,
$\payv$ is a continuous map $\theta\mapsto\payv_{\theta}(\cdot)$ from $\Theta$ to the space of continuous functions on $\pures$,
and the duality mapping is given by
\begin{equation}
\braket{\payv_{\theta}}{\nu}
	= \int_{\pures} \payv_{\theta}(\pure) \dd\nu(\pure)
	\quad
	\text{for all $\theta,\nu\in\Theta$}.
\end{equation}
In this framework, we introduce potential and monotone maps,
and we prove the convergence of the discrete-time dynamics of \acl{FP}.

We show that this study covers the case of nonatomic games with (player-dependent) compact action spaces.
Finally, we introduce a class of first-order mean-field games and prove that solutions of the \acl{SVI} correspond to the equilibria of the game in normal form where $\pures$ is the set of paths generated by the players.

\numberwithin{lemma}{section}		
\numberwithin{proposition}{section}		
\numberwithin{corollary}{section}		
\numberwithin{equation}{section}		
\appendix

\section{Properties of choice maps and the Fenchel coupling}
\label{app:mirror}

Our aim in this appendix is to provide some useful background and properties for the regularized choice map $\hreg$ and the Fenchel coupling.

To recall the basic setup, we assume below that $\hreg$ is a regularizer on $\strats$ in the sense of \cref{def:reg}.
The convex conjugate $\hconj\from\R^{\pures}\to\R$ of $\hreg$ is then defined as
\begin{equation}
\label{eq:hconj}
\hconj(\score)
	= \max_{\point\in\strats} \{ \braket{\score}{\point} - \hreg(\point) \}.
\end{equation}
Since $\strats$ is compact and $\hreg$ is strictly convex and continuous on $\strats$, the maximum in \eqref{eq:hconj} is attained at a unique point in $\strats$, so $\hconj(\score)$ is finite for all $\score\in\R^{\pures}$ \citep{BC17}.
Moreover, by standard results in convex analysis \citep[Chap.~26]{Roc70}, $\hconj$ is differentiable on $\R^{\pures}$ and its gradient satisfies the identity
\begin{equation}
\label{eq:dconj}
\nabla\hconj(\score)
	= \argmax_{\point\in\strats} \{ \braket{\score}{\point} - \hreg(\point) \}
	= \choice{}(\score)
\end{equation}
where $\choice{}\from\R^{\pures}\to\strats$ is the choice map induced by $\hreg$ (\cf~\cref{def:reg}).

The lemma below establishes the inverse differentiability property for $\hreg$, namely that $\subd\hreg = \choice{}^{-1}$:

\begin{lemma}
\label{lem:choice}
Let $\hreg$ be a regularizer on $\strats$.
Then, for all $\point\in\strats$ and all $\score\in\R^{\pures}$, we have:
\begin{equation}
\label{eq:links-choice}
\point = \choice{}(\score)
	\;\iff\;
	\score \in \subd\hreg(\point).
\end{equation}
If, in addition, $\hreg$ satisfies \eqref{eq:selection}, we also have
\begin{equation}
\label{eq:subinverse}
\braket{\nabla\hreg(\point)}{\point - \base}
	\leq \braket{\score}{\point - \base}
\end{equation}
for all $\base\in\strats$ and all $\score\in\subd\hreg(\point)$, $\point\in\dom\subd\hreg$.
\end{lemma}

\begin{proof}[Proof of \cref{lem:choice}]
To prove \eqref{eq:links-choice}, note that $\point$ solves \eqref{eq:dconj} if and only if $\score - \subd\hreg(\point) \ni 0$, \ie if and only if $\score\in\subd\hreg(\point)$.

For the inequality \eqref{eq:subinverse}, it suffices to show it holds for all $\base\in\dom\subd\hreg$ (by continuity).
To do so, let
\begin{equation}
\phi(t)
	= \hreg(\point + t(\base-\point))
	- \bracks{\hreg(\point) +  \braket{\score}{\point + t(\base-\point)}}.
\end{equation}
Since $\hreg$ is strictly convex and $\score\in\subd\hreg(\point)$ by \eqref{eq:links-choice}, it follows that $\phi(t)\geq0$ with equality if and only if $t=0$.
Moreover, note that $\psi(t) = \braket{\nabla\hreg(\point + t(\base-\point)) - \score}{\base - \point}$ is a continuous subgradient selection of $\phi$.
Given that $\phi$ and $\psi$ are both continuous on $[0,1]$, it follows a fortiori that $\phi$ is continuously differentiable and $\phi' = \psi$ on $[0,1]$.
Thus, with $\phi$ convex and $\phi(t) \geq 0 = \phi(0)$ for all $t\in[0,1]$, we conclude that $\phi'(0) = \braket{\nabla\hreg(\point) - \score}{\base - \point} \geq 0$, from which our claim follows.
\end{proof}

\begin{lemma}
\label{lem:Fench2norm}
Let $\hreg$ be a regularizer on $\strats$.
Then, for all $\base\in\strats$ and all $\score\in\R^{\pures}$, we have
\begin{equation}
\label{eq:Fench-pos}
\fench{\hreg}(\base,\score)
	\geq 0
	\quad
	\text{with equality if and only if $\base = \choice{}(\score)$}.
\end{equation}
If, in addition, if $\hreg$ satisfies \eqref{eq:strong},
we also have:
\begin{equation}
\label{eq:Fench2norm}
\fench{\hreg}(\base,\score)
	\geq \tfrac{\hstr}{2} \norm{\choice{}(\score) - \base}^{2}.
\end{equation}
\end{lemma}

\begin{proof}
Our first claim is a trivial consequence of the Fenchel-Young inequality and the strict convexity of $\hreg$.
For our second claim, let $\point = \choice{}(\score)$.
Then, by the definition \eqref{eq:hconj} of the convex conjugate of $\hreg$, we have:
\begin{align}
\label{eq:3point-develop}
\fench{\hreg}(\base,\score)
	&= \hreg(\base)
		+ \hconj(\score)
		- \braket{\score}{\base}
	\notag\\
	&= \hreg(\base)
		+ \braket{\score}{\choice{}(\score)}
		- \hreg(\choice{}(\score))
		- \braket{\score}{\base}
	\notag\\
	&= \hreg(\base)
		- \hreg(\point)
		- \braket{\score}{\base-\point}
\end{align}
Now, since $\hreg$ is $\hstr$-strongly convex, we also have:
\begin{equation}
\hreg((1-\coef)\strat + \coef\base) + \coef(1-\coef)(\hstr/2) \norm{\strat - \base}^{2}
	\leq (1-\coef)\hreg(\strat) + \coef\hreg(\base)
\end{equation}
and, with $\strat = \choice{}(\score)$, we get:
\begin{equation}
\braket{\score}{\strat} - \hreg(\strat)
	\geq \braket{\score}{(1-\coef)\strat + \coef\base}
		- \hreg((1-\coef)\strat + \coef\base).
\end{equation}
Thus, after rearranging, dividing by $\coef$, and letting $\coef\to0$, we obtain
\begin{equation}
\hreg(\base) - \hreg(\strat) - \braket{\score}{\base - \strat}
	\geq (\hstr/2) \norm{\strat - \base}^{2}
\end{equation}
and our assertion follows from \eqref{eq:3point-develop}.
\end{proof}

By virtue of this lemma, we obtain the following convergence criterion in terms of the Fenchel coupling:

\begin{corollary}
Let $\hreg$ be a regularizer on $\strats$ satisfying \eqref{eq:strong},
fix a base point $\base\in\strats$,
and
let $\score_{\run}$ be a sequence in $\R^{\pures}$.
Then, if $\fench{\hreg}(\base,\score_{\run}) \to 0$, we also have $\choice{}(\score_{\run}) \to \base$.
\end{corollary}

\begin{proof}
By \cref{lem:Fench2norm}, we have $\norm{\choice{}(\score_{\run}) - \base} \leq \sqrt{(2/\hstr) \fench{\hreg}(\base,\score_{\run})}$, so our claim follows.
\end{proof}

\section*{Acknowledgments}
{\small 
%
%
The authors acknowledge support by the COST Action CA16228 ``European Network for Game Theory'' (GAMENET),
the FMJH Program PGMO under grant Damper,
and
the French National Research Agency (ANR) in the framework of
the ``Investissements d'avenir'' program (ANR-15-IDEX-02),
the LabEx PERSYVAL (ANR-11-LABX-0025-01),
MIAI@Grenoble Alpes (ANR-19-P3IA-0003),
and the grants ORACLESS (ANR-16-CE33-0004) and ALIAS (ANR-19-CE48-0018-01).}

\bibliographystyle{plainfull}
\bibliography{bibtex/IEEEabrv,Bibliography}

\end{document}